\documentclass[10pt]{article}

\usepackage[utf8]{inputenc}
\usepackage[T1]{fontenc}
\usepackage[english]{babel} 
\usepackage{filecontents}
\usepackage{amsthm}
\usepackage{amsmath}
\usepackage{amssymb}
\usepackage{mathrsfs}
\usepackage{enumitem}
\usepackage{xcolor}
\usepackage{hyperref}
\hypersetup{colorlinks=true,linkbordercolor=red,linkcolor=blue,citecolor=magenta}
\usepackage[top=3cm, bottom=4cm, left=3cm , right=3cm]{geometry}
\usepackage{bm}
\usepackage{graphicx}
\usepackage{indentfirst}
\usepackage{multicol}
\usepackage[style=american]{csquotes}
\usepackage{wasysym}

\makeatletter
\newcommand{\subjclass}[2][1991]{%
  \let\@oldtitle\@title%
  \gdef\@title{\@oldtitle\footnotetext{#1 \emph{Mathematics subject classification.} #2}}%
}
\newcommand{\keywords}[1]{%
  \let\@@oldtitle\@title%
  \gdef\@title{\@@oldtitle\footnotetext{\emph{Keywords.} #1.}}%
}
\makeatother

\DeclareMathOperator \Tr {Tr}
\DeclareMathOperator{\Diag}{Diag}

\newcommand{\A}{\mathscr{A}}

\renewcommand{\L}{\mathcal{L}}
\newcommand{\M}{\mathscr{M}}
\newcommand{\N}{\mathbf{N}}
\renewcommand{\O}{O}
\renewcommand{\P}{\mathbf{P}}
\newcommand{\R}{\mathbf{R}}
\newcommand{\T}{\mathbf{T}}

\newcommand{\Z}{\mathbf{Z}}
\newcommand{\vep}{\epsilon}

\renewcommand{\d}{\mathrm{d}}

\newtheorem {prop} {Proposition}[section]
\newtheorem {lem} [prop] {Lemma}
\newtheorem{theorem}{Theorem}
\newtheorem{corol}  [prop]{Corollary}
\newtheorem{conj}[prop]{Conjecture}

\theoremstyle{definition}
\newtheorem{defi} [prop] {Definition}

\newtheorem{Hyp}[prop]{Assumptions}

\theoremstyle{remark}
\newtheorem{Rem} [prop] {Remark}

\numberwithin{equation}{section}

\makeatletter
\@addtoreset{prop}{section}
\makeatother

\usepackage[numbers, comma, square,sort]{natbib}

\title{Regular expansion for the characteristic exponent of a product of $2 \times 2$ random matrices}
\author{Benjamin Havret\thanks{
LPSM, UMR 8001, Université Paris Diderot, Sorbonne Paris Cité, F-75013 Paris France}}
\date{}

\begin{document}
\maketitle

\begin{abstract}
We consider a product of $2 \times 2$ random matrices which appears in the physics literature in the analysis of some 1D disordered models. These matrices depend on a parameter $\epsilon >0$ and on a positive random variable $Z$. Derrida and Hilhorst (J Phys A 16:2641, 1983, \S 3) conjecture that the corresponding characteristic exponent has a regular expansion with respect to $\vep$ up to --- and not further --- an order determined by the distribution of $Z$. We give a rigorous proof of that statement. We also study the singular term which breaks that expansion.
\newline
\noindent
\emph{AMS subject classification (2010 MSC) :} 82B44, 60B20,  37H15.
\end{abstract}

%

\section{Introduction}

Random matrix products appeared in the physics literature as a powerful tool to study disordered systems, ranging from Anderson model \cite{Matsuda_Ishii_70, Bougerol} to disordered harmonic chains \cite{Schmidt_57,Dyson_53} or disordered Ising model (discussed below). Among that wide range of models, the present work focuses on a very specific one, introduced by B. Derrida and H. Hilhorst in \cite{Derrida_Hilhorst_83} to study the strong interaction limit of a 1D disordered Ising model.

Let $(Z_n)$ be iid non-negative and non-deterministic random variables, with law $\mu$. For $\vep>0$, consider the matrices
\begin{equation}
\label{eq:matrixMvep}
M_{n,\vep} = \begin{pmatrix}
1 & \vep \\ \vep Z_n & Z_n
\end{pmatrix}.
\end{equation}
We will write $Z$ for a random variable with law $\mu$ and $M_\vep$ for the associated matrix. In fact, we will use $Z$ instead of $\mu$ to formulate our assumptions and results.
The (leading) Lyapunov exponent --- also called characteristic exponent --- is the growth rate of their product:
\begin{equation}
\label{eq:defi_lyap_Mnvep}
\L (\vep) =\L_Z(\vep) = \lim_{n \to + \infty} \frac{1}{n} \log \left \| M_{n,\vep} \cdots M_{1,\vep}\right \|.
\end{equation}
We will be particularly interested in the behaviour of $\L(\vep)$ in the limit $\vep \to 0$.

$2 \times 2$ matrices of the form \eqref{eq:matrixMvep} have appeared several times to express the free energy of the disordered 1D Ising model \cite{Derrida_Hilhorst_83, Calan_Luck,Nieuwenhuizen_Luck_86, Crisanti_Paladin_Vulpiani}, where the limit $\vep \to 0$ represents a regim of very strong interactions. It is also used in the celebrated work by B.~McCoy and T.~T.~Wu \cite{McCoy_Wu_article_68} to study a 2D Ising model with 1D disorder, as well as in a similar  model proposed by R. Shankar and G. Murthy \cite{Shankar_Murthy_87} which includes frustrated interactions.

From a mathematical point of view, a wide literature proposed to study these models and more general matrix products. One should cite the seminal work by H. Furstenberg et al. \cite{Furstenberg_Kesten_60,Furstenberg_63} and Oseledec's theorem \cite{Oseledec_68} (see \cite{Viana_2014} for a review). Looking at our own task, Furstenberg--Kesten theorem \cite{Furstenberg_Kesten_60} asserts that the limit~\eqref{eq:defi_lyap_Mnvep} exists almost surely and is deterministic, as long as $\mathrm{E}[\log_+ \|M_\vep\|]$ is finite (here $\mathrm{E}[\log_+ Z] <+ \infty$ suffices).
When $\vep$ vanishes, the matrix $M_{n, \vep}$ tends to a diagonal matrix and the Lyapunov exponent can be explicitly computed thanks to the law of large numbers: $\L(0) = \max(0, \mathrm{E}[\log Z])$. However, diagonal matrices are a degenerate case in the theory developed by H. Furstenberg et al. and one expects, in most cases, that $\vep \mapsto \L(\vep)$ is singular around~0.

It is worth stressing from now that the diagonal matrix $M_{n,0}$ is still random. Therefore we are not  in the framework of weak disorder limits such as \cite{Derrida_Zanon_88,Campanino_Klein_90,Sadel_SchulzBaldes_2010} in which the matrix $M_{n,0}$ is deterministic. The main reference paper for our analysis is rather \cite{Derrida_Hilhorst_83}.

\subsection{General conjecture and known results}

The present work is motivated by the recent mathematical progress by Genovese \emph{et al.} \cite{Genovese_Giacomin_Greenblatt_2017}. 
In this paper some physical  predictions, about the limiting behaviour of $\L(\vep)$ when $\vep$ vanishes, mainly stated in \cite{Derrida_Hilhorst_83} (see also \cite{Calan_Luck,Nieuwenhuizen_Luck_86,Crisanti_Paladin_Vulpiani}), are proven. However the physical predictions go beyond. We now first formulate these predictions in the form of conjectures, which detail the expected limiting behaviour of $\L(\vep)$, depending on the distribution of $Z$. Then we explain what has been proven and what our contribution is.

\begin{defi}
A real-valued random variable $\xi$ is said to be \emph{arithmetic} when there exists a constant $c >0$ such that $c \, \xi \in \Z \cup \{ \pm \infty \}$ almost surely.
\end{defi}

\begin{conj} 
\label{conj}
Assume that $\log Z$ is nonarithmetic.
\begin{enumerate}
\item Suppose in addition that there exists $\alpha \in(0,+\infty)$ such that $\mathrm{E}[Z^\alpha]=1$.
\begin{itemize}
\item If $\alpha \not \in \{1, 2, \ldots \}$, then, as $\vep$ goes to $0$,
\begin{equation}
\label{eq:conj_non_entier}
\L(\vep) = \sum_{k=1}^{\lfloor \alpha \rfloor} (-1)^{k+1} \ell_k \vep^{2 k} + (-1)^{\lceil \alpha \rceil+1} C_Z \vep^{2 \alpha} + o(\vep^{2 \alpha}),
\end{equation}
where, for $k \leqslant \lfloor \alpha \rfloor$, $\ell_k$ is a positive rational function of $\mathrm{E}[Z], \ldots, \mathrm{E}[Z^k]$; and $C_Z$ is a positive real number.
\item If $\alpha \in \{1, 2, \ldots \}$, then
\begin{equation}
\label{eq:conj_entier}
\L(\vep) = \sum_{k=1}^{\alpha -1} (-1)^{k+1} \ell_k \vep^{2 k} + (-1)^{\alpha+1} C_Z \vep^{2 \alpha} \log (1/\vep) + o\left( \vep^{2 \alpha} \log \vep \right),
\end{equation}
where the coefficients $(\ell_k)$ are the same positive rational functions of $Z$'s moments as before; and $C_Z$ is still a positive constant.
\end{itemize}
\item If $ \mathrm{E}[\log Z] = 0$ --- it is the ``$\alpha=0$'' case --- then
\begin{equation}
\label{eq:alpha=0}
\L(\vep) = \frac{C_Z}{\log (1/\vep)} + o\left((\log 1/\vep)^{-1}\right).
\end{equation}
\end{enumerate}
\end{conj}

The same references motivate further comments.

\begin{Rem}
The coefficients $(\ell_k)$ appearing in the conjecture can be computed recursively. For instance
\begin{equation}
\ell_1 = \frac{\mathrm{E}[Z]}{1-\mathrm{E}[Z]}, \qquad \ell_2 = \frac{(1+\mathrm{E}[Z])^2 \mathrm{E}[Z^2] + 2 \mathrm{E}[Z]^2 (1-\mathrm{E}[Z^2])}{2(1-\mathrm{E}[Z])^2(1-\mathrm{E}[Z^2])}.
\end{equation}
Precise recursive formulas will be derived in Section~\ref{sec:regular}. However it is not clear that a simple closed formula for $\ell_k$ can be derived. By contrast, apart from a few special situations, the calculation of the constant $C_Z$ is a very hard problem \cite[\S 4.2.1]{Crisanti_Paladin_Vulpiani}. On another note, in all the instances developed in Conjecture~\ref{conj}, the constant $C_Z$ should be replaced by a multiplicatively periodic function of $\vep$ if $\log Z$ is arithmetic. A precise computation of such a multiplicatively periodic function $C_Z$ is made in \cite{Derrida_Hilhorst_83} for a very specific (and singular) distribution of $Z$.
\end{Rem}

\begin{Rem}
\label{Rem:conj2}
We discuss in this remark the instances which are excluded by the conjecture. The conjecture actually covers almost all the cases where $\mathrm{E}[\log Z] \leqslant 0$ and $\mathrm{P}(Z >1) >0$, except the one discussed in the item~\ref{item:thmA:E[Zalpha]<1} of Remark~\ref{Rem:remarks_thmA}.
\begin{enumerate}
\item \label{item:Rem:conj2_parity} The case $\mathrm{E}[\log Z] >0$ (which corresponds to $\alpha <0$) boils down to $\mathrm{E}[\log Z]<0$ by factorizing $Z$ in the matrix $M_{\vep}$:
$
\L_Z(\vep)  = \mathrm{E}[\log Z] + \L_{1/Z} (\vep).
$
Similarly, by conjugating by the matrix $\Diag(-1,1)$, one observes that $\L$ is an even function: $\L(\vep) = \L(-\vep)$. It implies that the behaviour $\vep^{2 \alpha}$ is rather $|\vep|^{2 \alpha}$, so it is actually singular even when $\alpha$ is a half-integer.
\item If $Z \leqslant 1$ almost surely (that is ``$\alpha=+\infty$''), then $\L(\vep)$ admits a regular expansion with respect to $\vep^2$ up to any order. Is it smooth or analytic in a neighborhood of~0? The problem is still open, except if $Z \in [0, 1-\eta]$ almost surely, for some $\eta \in (0,1)$. If so then it is a consequence of a result by D. Ruelle \cite{Ruelle_79} that $\L(\vep)$ is a real analytic function of $\vep$ around~0.
\item That same theorem of D. Ruelle also ensures that the Lyapunov exponent $\L(\vep)$ is always an analytic function of $\vep$ on~$(0,+\infty)$.
\end{enumerate}
\end{Rem}

\medskip

Very little of Conjecture~\ref{conj} has been made mathematically rigorous. To our knowledge only \eqref{eq:conj_non_entier} has been successfully tackled \cite{Genovese_Giacomin_Greenblatt_2017}, and \emph{only} for $\alpha\in (0,1)$. When $\alpha\in (0,1)$, that is $\mathrm{E}[\log Z] <0$ and $\mathrm{E}[Z]>1$, the singularity $\vep^{2 \alpha}$ happens to be the leading behaviour of $\L(\vep)$ and~\eqref{eq:conj_non_entier} takes the form:
\begin{equation}
\label{eq:alpha_in(0,1)}
\L(\vep) \underset{\vep \searrow 0}{\sim} C_Z \, \vep^{2 \alpha}.
\end{equation} 
This simplifies in a substantial way the analysis:  Derrida and Hilhorst in \cite{Derrida_Hilhorst_83} (see also \cite{Derrida_Hilhorst_83, Calan_Luck,Nieuwenhuizen_Luck_86}) explicitly give, based on a two scale argument, a probability measure that is expected to be close, when~$\vep$ is small, to the invariant probability for the action of $M_\vep$ on the projective space $\P^1(\R)$ (that is, the distribution of $X_\vep$, in the next paragraph's notations). Then, they use this probability to compute the Lyapunov exponent. This two scale analysis is made rigourous by G. Genovese \emph{et al.} \cite{Genovese_Giacomin_Greenblatt_2017}, who show that this probability measure is indeed close to the invariant measure in a suitable norm, and this control is sufficiently strong to yield precisely \eqref{eq:alpha_in(0,1)}. It appears to be rather challenging to follow the same steps 
for $\alpha \geqslant 1$: the guess for the invariant probability would have to be tuned to yield the $\lfloor \alpha \rfloor$ terms of the regular expansion and the singular $\vep^{2\alpha}$ term. Even at a heuristic level, such a construction is lacking. Note, in particular, that in \cite{Derrida_Hilhorst_83}, the $\alpha\geqslant 1$ case is treated in a expedite way, without reference to the invariant probability, and without capturing the singularity $\vep^{2 \alpha}$.  

On the other hand, a weak disorder limit of the model has been investigated. In this limit, the product of random matrices becomes a stochastic differential equation system. An exactly solvable structure emerges from that SDE and the analog of \eqref{eq:conj_non_entier} and \eqref{eq:conj_entier} has been shown to hold (see \cite{Grabsch_Texier_Tourigny_2014} for the case $\alpha \in (0,2)$ and \cite{Comets_Giacomin_Greenblatt_2017} for the general case). As pointed out in \cite{Comets_Giacomin_Greenblatt_2017}, it is rather remarkable that the structure of \eqref{eq:conj_non_entier} and \eqref{eq:conj_entier} holds also in the weak disorder limit and this appears to be a rather deep fact. Nonetheless, the fact that the conjecture holds in the weak disorder limit is far from being a mathematical proof of the conjecture for products of matrices.

\medskip

The main aim of our work is to approach~\eqref{eq:conj_non_entier} and~\eqref{eq:conj_entier}. Our results are the following.

\begin{enumerate}
\item  $\L(\vep)$ admits a regular expansion in powers of $\vep^2$, up to order $\vep^{2\lfloor \alpha \rfloor}$, or $\vep^{2(\alpha-1)}$ in the integer case (that is the regular part of~\eqref{eq:conj_non_entier} and~\eqref{eq:conj_entier}).

\item We prove that the next order term after this regular part, call it $R(\vep)$, satisfies, as $\vep$ goes to $0$, for instance in the non-integer case
\begin{equation}
\vep^{2 \lfloor \alpha \rfloor} \ll R(\vep) \ll \vep^{2 (\lfloor \alpha \rfloor +1)}.
\end{equation} 
Since $\vep \mapsto \L(\vep)$ is an even function, only even powers of $\vep$ are non-singular (see Remark~\ref{Rem:conj2} item~\ref{item:Rem:conj2_parity}). Hence, $R(\vep)$ is necessarily singular. A quantitative and explicit control on this term is given, but it falls short to prove the expected $\vep^{2 \alpha}$ behavior of~\eqref{eq:conj_non_entier}.
\end{enumerate}

\subsection{Assumptions and main result}

We will work under the following assumptions, supposed to be satisfied in the whole paper.

\begin{Hyp}
\label{hyp}
The random variable $Z$ is positive, non-deterministic, and 
\begin{enumerate}[label=(\alph*)]
\item $\mathrm{E}[\log Z] <0$ (can be $-\infty$);
\item There exists $\delta >0$ such that $\mathrm{E}[Z^\delta]<+\infty$.
\end{enumerate}
\end{Hyp}

\noindent
Introduce
\begin{equation}
\A = \{ \gamma \in [0, + \infty] \text{  such that  } \mathrm{E}[Z^\gamma] <1 \},
\end{equation}
and
\begin{equation}
\alpha = \sup \A \in (0,+\infty].
\end{equation}
The Assumptions~\ref{hyp}, together with a convexity argument, ensure that $\A$ is an interval of positive length. Note that $\alpha =+  \infty$ if and only if $Z \leqslant 1$ almost surely. In any case $\A$ takes the following form: either $\mathrm{E}[Z^\alpha]=1$ and then $\A = (0, \alpha )$, or $\mathrm{E}[Z^\alpha]<1$, and then $\A = (0, \alpha ]$. In the latter case, necessarily, $\mathrm{E}[Z^\gamma]=+\infty$ for every $\gamma > \alpha$.
Here is the main result of this work.

\begin{theorem}
\label{The:thm_concret}
There exist positive coefficients $(\ell_k)$, where $\ell_k$ is a rational function of the moments $\mathrm{E}[Z], \ldots, \mathrm{E}[Z^k]$, such that the following expansions hold, as $\vep$ goes to~0.
\begin{enumerate}
\item If $\alpha = + \infty$ (\emph{i.e.}, if $Z\leqslant 1$ a.s.), then for every $K \geqslant 0$,
\begin{equation}
\label{eq:expansionLvep_the_alpha_infini}
\L (\vep) = \sum_{k=1}^{K} (-1)^{k+1}  \ell_k \vep^{2 k} +  \O(\vep^{2(K+1)}).
\end{equation}
\item If $\alpha \in \{1, 2, \ldots \}$ and if $\mathrm{E}[Z^\alpha]=1$, then
 \begin{equation}
\label{eq:expansionLvep_The_entier}
\L (\vep) = \sum_{k=1}^{\alpha -1} (-1)^{k+1}  \ell_k \vep^{2 k} +  (-1)^{ \alpha  +1} R(\vep),
\end{equation}
where $R(\vep)$ is nonnegative and
\begin{equation}
\label{eq:The_encadrement_Rvep_entier}
\vep^{2 \alpha} \ll R(\vep) \leqslant C \vep^{2 \alpha} \log (1/\vep),
\end{equation}
for some $C>0$. The lower bound can be improved if, in addition, $Z$ has a bounded support, to obtain, for some $C \geqslant c>0$, the sharper estimate
\begin{equation}
\label{eq:The_encadrement_Rvep_entier_Z_borne}
c  \leqslant \frac{R(\vep)}{\vep^{2 \alpha} \log (1/\vep)} \leqslant C.
\end{equation}
\item If $\alpha \in (0, + \infty) \backslash \{1, 2, \ldots\}$ and if there exists $\gamma>\alpha$ such that $\mathrm{E}[Z^\gamma]$ is finite, then
 \begin{equation}
 \label{eq:expansionLvep_The_nonentier}
\L (\vep) = \sum_{k=1}^{\lfloor \alpha \rfloor} (-1)^{k+1}  \ell_k \vep^{2 k} +  (-1)^{\lceil \alpha \rceil +1} R(\vep),
\end{equation}
where $R(\vep)$ is nonnegative and
\begin{equation}
\label{eq:The_encadrement_Rvep_nonentier}
\vep^{2 \lceil \alpha \rceil} \ll R(\vep) \leqslant C \vep^{2 \alpha},
\end{equation}
for some $C>0$. The lower bound can be improved if, in addition, $Z$ has a bounded support: in that case, there exists $\theta \in (\alpha, \lceil \alpha \rceil)$ and $c>0$ such that $R(\vep) \geqslant c \vep^{2 \theta}$.
\end{enumerate}
\end{theorem}

\begin{Rem}
\label{Rem:remarks_thmA}
\begin{enumerate}
\item The constant $\theta$ is explicit:
$
\theta = \lceil \alpha \rceil - \frac{\log \mathrm{E} [Z^{\lceil  \alpha \rceil}]}{\log \|Z\|_\infty}.
$
\item When $\alpha$ is finite, the lower bounds of the error in~\eqref{eq:The_encadrement_Rvep_entier} and~\eqref{eq:The_encadrement_Rvep_nonentier} assert in particular that the regular expansions~\eqref{eq:expansionLvep_The_entier} and~\eqref{eq:expansionLvep_The_nonentier} cannot be continued beyond $K=\lceil \alpha \rceil -1$: $\L(\vep)$ is singular.
\item \label{item:rem_thmA_ZlogZ} When $\alpha$ is not an integer, the assumption ``there exists $\gamma > \alpha $ such that $\mathrm{E}[Z^\gamma]$ is finite'' can be replaced by the weaker assumption ``$\mathrm{E}[Z^\alpha \log_+ Z] <+\infty$'' (see Remark~\ref{Rem:proof_E[Zalpha]<1_EZlogZ}).
\item \label{item:thmA:E[Zalpha]<1} Suppose that $\alpha$ is finite and $\mathrm{E}[Z^\alpha]<1$ (and $\mathrm{E}[Z^\gamma] = + \infty$ for every $\gamma > \alpha$). Whether $\alpha$ is an integer or not, under some technical assumptions on the distribution of $Z$, the Lyapunov exponent is slightly regularized (see Remark~\ref{Rem:proof_E[Zalpha]<1_EZlogZ} for a sketch of proof):
 \begin{equation}
\label{eq:expansionLvep_The_entier_pathologique}
\L (\vep) = \sum_{k=1}^{\lfloor \alpha \rfloor} (-1)^{k+1}  \ell_k \vep^{2 k} + (-1)^{\lfloor \alpha \rfloor +1} R(\vep), \qquad \vep^{2 (\lfloor \alpha \rfloor +1)} \ll R(\vep) \ll \vep^{2 \alpha}.
\end{equation}
\item  We mention that, when $\alpha \in(0,1)$, Theorem~\ref{The:thm_concret} gives an estimate of $\L(\vep)$ which is rough, and of course strongly weaker than~\cite{Genovese_Giacomin_Greenblatt_2017}.
\end{enumerate}
\end{Rem}

\subsection{Strategy of the proof and structure of the paper}

A classical result in the theory of product of random matrices ensures that the Lyapunov exponent can be written
\begin{equation}
\label{eq:expression_lyapunov}
\L(\vep) = \mathrm{E}[\log(1+\vep^2 X_\vep)],
\end{equation}
where $X_\vep$ is an invariant measure for the random transformation, on $[0, +\infty)$,
$
x \mapsto Z \frac{1+ x}{1+ \vep^2 x}.
$
In other words it satisfies
\begin{equation}
\label{eq:intro_invarianceXvep}
 X_\vep \overset{\normalfont{(d)}}{=} Z \frac{1+ X_\vep}{1+ \vep^2 X_\vep},
\end{equation}
where $Z$ is independent of $X_\vep$ (on the right hand side).
Existence and uniqueness of such a random variable $X_\vep$ will be justified in Section~\ref{sec:defXvep}, as well as formula~\eqref{eq:expression_lyapunov}. A very useful uniform stochastic dominance of the random variables $(X_\vep)_{\vep >0}$ will also be proved.

From that point on, the work will only be based on formula~\eqref{eq:expression_lyapunov} for the Lyapunov exponent and the fixed point equation~\eqref{eq:intro_invarianceXvep}. Thanks to the former, the problem will readily boil down to studying $X_\vep$'s moments. That study can be split into two subproblems. We will know since Section~\ref{sec:defXvep} which ones of $X_\vep$'s moments are bounded as $\vep$ goes to~0 and which diverge. The two subproblems then are:
\begin{itemize}
\item Deriving a regular expansion for $X_\vep$'s bounded moments, involving an error in terms of a divergent moment of $X_\vep$ (Sections~\ref{sec:regular} and~\ref{sec:error});
\item Estimating the divergence speed of $X_\vep$'s unbounded moments (Section~\ref{section:Xvep^K+1}).
\end{itemize}

The former point is addressed in Section~\ref{sec:regular}. The analysis is based on a bootstrap procedure, based on recursive uses of the fixed point equation~\eqref{eq:intro_invarianceXvep}. It gives more and more precise expansions of these moments. 
Eventually, it will provide the regular expansion~\eqref{eq:expansionLvep_The_entier} or~\eqref{eq:expansionLvep_The_nonentier} with an upper bound on the error~$R(\vep)$, involving a divergent moment of $X_\vep$. That work will be generalized in the appendix~\ref{appendix:generalization}, for matrices of size $d$, with more general entries.

That same strategy, using a bootstrap procedure to obtain a more and more precise estimate of $X_\vep$'s moments, can also provide a lower bound on the error, involving a divergent truncated moment of $X_\vep$: Section~\ref{sec:error} will be devoted to that analysis.

At the end of these sections, the following theorem will be proved, which, unlike Theorem~\ref{The:thm_concret}, does not require any extra assumption on $Z$ (apart from Assumptions~\ref{hyp}).

\begin{theorem} 
\label{The:dvpt_avec_erreur_general}
Fix $B>0$, and an integer $K \in \A \cup\{0\}$. One has, for all $\vep >0$,
\begin{equation}
\label{eq:expansionLvep_the}
\L (\vep) = \sum_{k=1}^{K} (-1)^{k+1}  \ell_k \vep^{2 k} +  (-1)^{K+2} R_{K}(\vep),
\end{equation}
where, for all $\beta \in (K,K+1]$, and for some positive constants $c$ and $C_\beta$,
\begin{equation}
\label{eq:The_encadrement_Rvep}
c \vep^{2 (K+1)} \mathrm{E}[X_\vep^{K+1} \mathbf{1}_{ \vep^2 X_\vep \leqslant B}] \leqslant  R_{K}(\vep) \leqslant C_\beta \vep^{2 \beta} \mathrm{E}[X_\vep^{\beta}].
\end{equation}
\end{theorem}

\begin{Rem}
The coefficients $(\ell_k)$ are the same as in Theorem~\ref{The:thm_concret}.
The neat thing about that theorem is that, unlike the lower bound~\eqref{eq:The_encadrement_Rvep_nonentier} of Theorem~\ref{The:thm_concret}, the estimate~\eqref{eq:The_encadrement_Rvep} should be ``sharp'' in the following sense.  If one proves that, as $\vep$ goes to $0$, $\mathrm{P}( X_\vep \geqslant c \vep^{-2} ) \geqslant C \vep^{2 \alpha}$
for some positive constants $c$ and $C$ (the precise analysis of $M_\vep$'s invariant measure conducted in \cite{Genovese_Giacomin_Greenblatt_2017} provides such an estimate when $\alpha \in (0,1)$)  then~\eqref{eq:The_encadrement_Rvep} becomes $c \vep^{2 \alpha} \leqslant R_K(\vep) \leqslant C \vep^{ 2 \alpha}$ (with a log correction if $\alpha$ is an integer). It is the good order of $\vep$ predicted by Conjecture~\ref{conj}. Without such an estimate, \eqref{eq:The_encadrement_Rvep} is not satisfactory yet for it is not explicit enough. 
\end{Rem}

To obtain the explicit bounds given in Theorem~\ref{The:thm_concret}, a study of the divergence speed of $X_\vep$'s divergent moments is needed. It is conducted in Section~\ref{section:Xvep^K+1}. The derivation of upper bounds is based a stochastic dominance found in Section~\ref{sec:defXvep} (namely $X_\vep \preccurlyeq X_0$), and on renewal theory results describing the limiting behaviour of the tail of $X_0$. The lower bounds are only derived when $Z$ is bounded. The analysis is again based on a recursive use of the fixed point equation~\eqref{eq:intro_invarianceXvep}. It is the point where the sharpness of the lower bound~\eqref{eq:The_encadrement_Rvep} of the singularity is lost. 
Theorem~\ref{The:thm_concret} is proved at the end of Section~\ref{section:Xvep^K+1}.

\numberwithin{prop}{section}

\section{Existence and first properties of the invariant measure \texorpdfstring{$X_\vep$}{}}
\label{sec:defXvep}

In this section we prove the existence of the random variables $X_\vep$ and derive formula~\eqref{eq:expression_lyapunov}.
A first result on $X_\vep$'s moments is also proved: it spells out which moments of $X_\vep$ are bounded as $\vep$ goes to~0 and which diverge.

We start by introducing an invariant measure of the random matrix $M_0$ ($\vep=0$). It will play a central role to define the random variables $X_\vep$ and control their moments. First I need to fix a notation for the stochastic dominance.

\begin{defi}
The \emph{stochastic dominance} will be denoted by $\preccurlyeq$. Formally, if $X$ and $Y$ are two real-valued random variables, $X \preccurlyeq Y$ means that 
$\mathrm{P}(X \geqslant x) \leqslant \mathrm{P}(Y \geqslant x)$ for every $x \in \R$.
Equivalently, there exist two copies $\tilde X$ and $\tilde Y$, of $X$ and $Y$ respectively, such that $\tilde X \leqslant \tilde Y$ almost surely.
\end{defi}

\begin{lem}
\label{lem:X0}
Fix a sequence $(Z_n)$ of iid copies of $Z$.
The series
\begin{equation}
X_0 = \sum_{n = 1}^{+\infty} Z_1 \cdots Z_n
\end{equation}
converges almost surely. It is the unique random variable (in distribution) satisfying
\begin{equation}
\label{eq:egalité_loi_X0}
X_0 \overset{\normalfont{(d)}}{=} Z (1+X_0),
\end{equation}
with $Z$ independent of $X_0$.
Moreover $\mathrm{E}[\log_+ X_0]$ is finite; and for every $\gamma > 0$,
\begin{equation}
\mathrm{E}[X_0^{\gamma}]< + \infty \qquad \text{ if and only if } \qquad \mathrm{E}[Z^\gamma]<1.
\end{equation}
\end{lem}

\begin{proof} Recall that $\mathrm{E}[\log Z] <0$. The almost sure convergence of the series follows from the law of large numbers, whereby
\begin{equation}
\label{eq:LGN}
Z_1 \cdots Z_n = e^{ n  \mathrm{E} [\log Z] + o(n)} \qquad \textrm{ as } n \to + \infty.
\end{equation}
Of course
\begin{equation}
X_0 = \sum_{n = 1}^{+\infty} Z_1 \cdots Z_n = Z_1 \left( 1 + \sum_{n = 2}^{+\infty} Z_2 \cdots Z_n \right)
\end{equation}
satisfies the identity~\eqref{eq:egalité_loi_X0}. 
Let's turn to the uniqueness. If $\tilde X_0$ is another random variable satisfying~\eqref{eq:egalité_loi_X0}, then, applying this identity $N$ times we get
\begin{equation}
\tilde X_0 \overset{\normalfont{(d)}}{=} \sum_{n = 1}^{N} Z_1 \cdots Z_n + Z_1 \cdots Z_N \tilde X_0,
\end{equation}
where $Z_1, \ldots, Z_N$ are iid copies of $Z$, independent of $\tilde X_0$.
With~\eqref{eq:LGN}, the last term vanishes (in distribution) as $N$ goes to $+\infty$, whereas the first sum converges monotonically towards $X_0$. So eventually, $\tilde X_0 \overset{\normalfont{(d)}}{=} X_0$. The uniqueness is proved.

Now fix $\gamma > 0$ such that $\mathrm{E}[Z^{\gamma}]<1$. We want to prove that $\mathrm{E}[X_0^{\gamma}] $ is finite. If $\gamma \geqslant 1$ we use Minkovsky's inequality:
\begin{equation}
\label{eq:minkowski1}
\mathrm{E}[X_0^{\gamma}]^{1/{\gamma}} \leqslant \sum_{n=1}^{+\infty} \mathrm{E}[(Z_1 \cdots Z_n)^\gamma]^{1/\gamma} = \sum_{n=0}^{+\infty} \mathrm{E}[Z^\gamma]^{n/\gamma}.
\end{equation}
Thus $\mathrm{E}[X_0^\gamma]$ is finite.
On the other hand, if $\gamma \in (0,1)$, then for all $x,y \geqslant 0$, $(x+y)^{\gamma} \leqslant x^{\gamma} + y^{\gamma}$.
So
\begin{equation}
\label{eq:minkowski2}
\mathrm{E}[X_0^{\gamma}] \leqslant \sum_{n=0}^{+\infty} \mathrm{E}[Z^{\gamma}]^{n},
\end{equation}
which is again finite.
Now, if $\mathrm{E}[Z^{\gamma}]\geqslant 1$, then with the identity~\eqref{eq:egalité_loi_X0},
\begin{equation}
\mathrm{E}[X_0^{\gamma}] = \mathrm{E}[Z^{\gamma}] \mathrm{E}[(1+X_0)^{\gamma}] \geqslant \mathrm{E}[(1+X_0)^{\gamma}], 
\end{equation}
which can hold only if $\mathrm{E}[X_0^{\gamma}]=+\infty$ (or $\gamma =0$).
Eventually, pick $\gamma \in \A$ so that $\mathrm{E}[Z^\gamma] <1$. With the foregoing, we then know that $\mathrm{E}[X_0^\gamma]<+\infty$. Thus, by Jensen's inequality, $\mathrm{E}[\log_+ X_0]$ is finite.
\end{proof}

The next lemma provides the existence of the random variables $X_\vep$ and the desired formula for the Lyapunov exponent.

\begin{lem}
\label{lem:existence_Xvep}
For all $\vep >0$, there exists a non-negative random variable $X_\vep$, unique in distribution, such that
\begin{equation}
\label{eq:invariance_Xvep}
X_\vep \overset{\normalfont{(d)}}{=} Z \frac{1+ X_\vep}{1+ \vep^2 X_\vep},
\end{equation}
with $Z$ independent of $X_\vep$. Moreover, for every $\vep >0$, 
$Z \preccurlyeq X_\vep \preccurlyeq X_0$.
Furthermore,
\begin{equation}
\label{eq:lem:existence_Xvep_Lvep}
\L(\vep) = \mathrm{E}[\log (1 + \vep^2 X_\vep)],
\end{equation}
and $\L(\vep)$ is also the growth rate of the entries of $M_{n,\vep} \cdots M_{1,\vep}$: for every $\vec x, \vec y \in \R^2$ with nonnegative entries,
\begin{equation}
\label{eq:lem:convergence_entries}
\frac{1}{n} \log \left \langle \vec x , M_{n,\vep} \cdots M_{1,\vep}  \vec y \right \rangle \underset{n \to + \infty} {\longrightarrow} \L(\vep) \qquad \text{a.s. and in } L^1.
\end{equation}
\end{lem}

\begin{Rem}
There could be other distributions, supported on $\R$, satisfying~\eqref{eq:invariance_Xvep}. We only claim uniqueness for non-negative invariant measure. However, if $Z$ does not have a finite support, then one can prove, using classical results of products of random matrices (see \cite[Chapter 3]{Bougerol}), that there exists a unique invariant measure on $\overline{\R}$. With Lemma~\ref{lem:existence_Xvep}, we know that it must be supported on $\R_+$.
\end{Rem}

In what follows, $X_\vep$ will always denote the unique non-negative invariant random variable of Lemma~\ref{lem:existence_Xvep}.

\begin{proof}
We begin with the proof of the existence, for which we use a standard procedure. Fix an iid sequence $(Z_n)$ of copies of $Z$, set $x_0 = 0$ and define recursively the random variables
\begin{equation}
\label{eq:xn+1=f(xn)}
x_{n+1} = Z_{n+1} \frac{1+ x_n}{1+ \vep^2 x_n}.
\end{equation}
Denote by $\nu_n$ the distribution of $x_n$ and consider the measure $\rho_N = \frac{1}{N} \sum_{n=0}^{N-1} \nu_n$.
Observe that for any $n \geqslant 0$, $x_n$ is nonnegative and $x_{n+1} \leqslant Z_{n+1} (1+ x_n)$.
Thus, by an easy induction,
\begin{equation}
0 \leqslant   x_n \leqslant \sum_{k=0}^{n-1} Z_n \cdots Z_{n-k} \preccurlyeq X_0:
\end{equation}
the random variables $x_n$ are uniformly bounded by $X_0$. Consequently the sequence $(\rho_N)$ is tight. Pick a limit point $\rho_\infty$ of that sequence and fix a random variable $X_\vep$ with distribution $\rho_\infty$. The limit distribution $\rho_\infty$ must be invariant under the random transformation~\eqref{eq:xn+1=f(xn)}. In other words it must satisfy~\eqref{eq:invariance_Xvep}.
The existence of an invariant measure supported on $\R_+$ is proved. Incidentally we obtained $X_\vep \preccurlyeq X_0$. As for the stochastic lower bound $X_\vep \succcurlyeq Z$, it directly follows from the identity~\eqref{eq:invariance_Xvep}.

To deal with the uniqueness, assume that $X^{(0)}_\vep$ and $Y^{(0)}_\vep$ are two such random variables and fix an iid sequence $(Z_n)$ of copies of $Z$, independent of $X^{(0)}_\vep$ and $Y^{(0)}_\vep$. We introduce, for $n \geqslant 0$,
\begin{equation}
\label{eq:def_Xvepn_Yvepn}
X_\vep^{(n+1)} = Z_{n+1} \frac{1+ X^{(n)}_\vep}{1+ \vep^2 X^{(n)}_\vep}, \qquad Y_\vep^{(n+1)} = Z_{n+1} \frac{1+Y^{(n)}_\vep}{1+ \vep^2 Y^{(n)}_\vep}.
\end{equation}
Observe that, almost surely,
\begin{equation}
| X^{(n+1)}_\vep- Y^{(n+1)}_\vep | = Z_{n+1} \frac{(1-\vep^2)|X^{(n)}_\vep-Y^{(n)}_\vep|}{(1+ \vep^2 X^{(n)}_\vep)(1+ \vep^2 Y^{(n)}_\vep)} \leqslant Z_{n+1} |X^{(n)}_\vep-Y^{(n)}_\vep|.
\end{equation}
Thus, with~\eqref{eq:LGN}, $| X^{(n)}_\vep- Y^{(n)}_\vep |$ vanishes almost surely as $n$ goes to $+ \infty$.
On the other hand, note that, with the construction~\eqref{eq:def_Xvepn_Yvepn}, for all $n \geqslant 0$, $X^{(n)}_\vep \overset{\normalfont{(d)}}{=} X^{(0)}_\vep$ and  $Y^{(n)}_\vep \overset{\normalfont{(d)}}{=} Y^{(0)}_\vep$.
The uniqueness follows. Then $\rho_N$ actually converges (without extraction) towards $X_\vep$'s distribution.

We are left with the proof of formula~\eqref{eq:lem:existence_Xvep_Lvep}. Thanks to a result by H. Hennion \cite{Hennion_97}, since $M_\vep$'s entries are positive, the convergence~\eqref{eq:lem:convergence_entries} holds. On the other hand, for every $n \geqslant 0$,
\begin{equation}
M_{n,\vep}  \cdots M_{1,\vep} \begin{pmatrix}
1 \\ \vep X_\vep^{(0)} \end{pmatrix}
 = \left[ \prod_{k=0}^{n-1} (1+\vep^2 X_\vep^{(k)}) \right] \begin{pmatrix}
1 \\ \vep X_\vep^{(n)} \end{pmatrix}
\end{equation}
So, by taking the $\log$ and the expectation,
\begin{equation}
\frac{1}{n} \mathrm{E}[ \log \| M_{n,\vep}  \cdots M_{1,\vep} {}^t (1, \vep X_\vep^{(0)}) \| ] = 
\mathrm{E}[\log(1+ \vep^2 X_\vep)] + \frac{1}{n} \mathrm{E}[\log \| (1, \vep X_\vep)\|].
\end{equation}
Since $\mathrm{E}[\log_+ X_\vep] \leqslant \mathrm{E}[\log_+ X_0]$ is finite (Lemma~\ref{lem:X0}), the last term vanishes as $n$ goes to $+ \infty$. On the other hand, one has, for every $n \geqslant 0$,
\begin{equation}
(1,0) M_{n,\vep} \cdots M_{1,\vep} {}^t(1,0)
\leqslant \|M_{n,\vep} \cdots M_{1,\vep} {}^t (1, \vep X_\vep^{(0)})\| 
\leqslant \|M_{n,\vep} \cdots M_{1,\vep}\|(1 + X_\vep^{(0)}).
\end{equation}
Since we know that both the lower and upper bounds goes to $\L(\vep)$ (after taking $\log$ and expectation) as $n$ goes to $+\infty$, almost surely and in $L^1$, we get the result.
\end{proof}

\begin{Rem}Formula~\eqref{eq:lem:existence_Xvep_Lvep} can also be proved with a classical result by H. Furstenberg and Y. Kifer \cite[Corollary of Theorem 3.10]{Furstenberg_Kifer_83}, which gives an explicit formula for the Lyapunov exponent in terms of invariant measures as soon as $M$ is an invertible random matrix of size $d \times d$ with no deterministic proper invariant subspace. We could also have used the convergence $\rho_n \to \L(X_\vep)$ to prove~\eqref{eq:lem:existence_Xvep_Lvep} and~\eqref{eq:lem:convergence_entries} without using H. Hennion's results.
\end{Rem}

\begin{Rem}
If one notes that the map $\vep \mapsto \frac{1+x}{1+\vep^2 x}$.
is monotone, one obtains, with the previous construction, that the random variables $X_\vep$ are stochastically decreasing with $\vep$: for all $\vep' \geqslant \vep >0$ one has
$X_{\vep'}  \preccurlyeq X_{\vep}  \preccurlyeq X_0$.
\end{Rem}

\begin{lem}
\label{lem:convergence_en_loi_Xvep}
$X_\vep \to X_0$ in distribution when $\vep \to 0$.
\end{lem}

\begin{proof}
The stochastic dominance $X_\vep  \preccurlyeq X_0$ ensures that the family of random variables $(X_\vep)_{\vep>0}$ is tight. Consider a limit point $\tilde X_0$ of $X_\vep$ as $\vep$ goes to $0$. Since $X_\vep$ satisfies the identity~\eqref{eq:invariance_Xvep}, the limit point $\tilde X_0$ must satisfy $\tilde X_0 \overset{\normalfont{(d)}}{=} Z (1+ \tilde X_0)$.
That means, using Lemma~\ref{lem:X0}, that $X_0$ is the only possible limit point of $X_\vep$ as $\vep$ goes to $0$. The convergence of $X_\vep$ towards $X_0$ (in distribution) follows.
\end{proof}

Using classical integration theorems, one readily obtains the following limiting behaviour of $X_\vep$'s moments, or truncated moments, which will be needed in the proof of Theorem~\ref{The:thm_concret}.

\begin{corol}
\label{corol:moments_Xvep_premier_resultat}
For any $\gamma >0$,
\begin{enumerate}
\item If $\mathrm{E}[Z^{\gamma}]<1$ then, as $\vep$ goes to~0, $\mathrm{E}[X_\vep^\gamma] = \O(1)$.
\item If $\mathrm{E}[Z^{\gamma}] \geqslant 1$ then for any $B >0$,
\begin{equation}
\mathrm{E}\left[X_\vep^{\gamma} \mathbf{1}_{\vep^2 X_\vep \leqslant B}\right] \underset{\vep \to 0} {\longrightarrow} + \infty.
\end{equation}
\end{enumerate}
\end{corol}

\begin{proof}
Recall that $\mathrm{E}[X_0^\gamma]$ is finite if and only if $\mathrm{E}[Z^\gamma]<1$ (Lemma~\ref{lem:X0}). With the stochastic dominance $X_\vep \preccurlyeq X_0$ provided by Lemma~\ref{lem:existence_Xvep}, we get $\mathrm{E}[X_\vep^\gamma ] = \O(1)$ when $\mathrm{E}[Z^\gamma] <1$. On the other hand if $\mathrm{E}[Z^\gamma]\geqslant 1$, then $\mathrm{E}[X_0^\gamma]=+\infty$ (Lemma~\ref{lem:X0}).
Since $X_\vep \preccurlyeq X_0$, one can pick representatives $\tilde X_\vep$ and $\tilde X_0$ such that $\tilde X_\vep \leqslant \tilde X_0$ almost surely. It gives the lower bound
\begin{equation}
\mathrm{E}\left[X_\vep^{\gamma} \mathbf{1}_{\vep^2 X_\vep \leqslant B}\right] \geqslant \mathrm{E}\left[\tilde X_\vep^{\gamma} \mathbf{1}_{\vep^2 \tilde X_0 \leqslant B}\right]
\end{equation}
for any $B >0$. By Fatou's lemma and the convergence in distribution provided by Lemma~\ref{lem:convergence_en_loi_Xvep}, the latter lower bound goes to $+\infty$ as $\vep$ goes to $0$.
\end{proof}

\section{Regular expansion (Theorem~\ref{The:dvpt_avec_erreur_general}: upper bound)}
\label{sec:regular}

In this section we prove the existence of a regular expansion for the Lyapunov exponent $\L(\vep)$. We also  lay out the method, which will be used twice more: for the generalization of this result in Appendix~\ref{appendix:generalization} and in Section~\ref{sec:error} to obtain the lower bound of the error. It is based on the study of a regular expansion for the moments of $X_\vep$ which are bounded as $\vep $ goes to~0.
Let us first state the main result of the section.

\begin{prop}
\label{prop:dvpt_lyapunov}
Pick an integer $K \in \A \cup\{0\}$, and fix $\beta \in [K,K+1]$.
The following expansion holds when $\vep$ goes to $0$,
\begin{equation}
\label{eq:prop:dvpt_lyapunov}
\L(\vep) = \sum_{k=1}^K (-1)^{k+1} \ell_k \vep^{2 k} + \O(\vep^{2 \beta}\mathrm{E}[X_\vep^\beta]),
\end{equation}
where, for $k \leqslant K$, the coefficient $\ell_k$ is a positive rational function of $\mathrm{E}[Z], \ldots, \mathrm{E}[Z^k]$.
\end{prop}

\begin{Rem}
\label{Rem:O(xalpha_xbeta)}
With some extra effort, $\O(\vep^{2 \beta}\mathrm{E}[X_\vep^\beta])$ can be replaced by
$\O(\mathrm{E}[(\vep^2 X_\vep)^{\beta_1} \wedge (\vep^2 X_\vep)^{\beta_2}])$ for any $\beta_1, \beta_1 \in [K,K+1]$. It will only be needed to explain some generalizations discussed in Remark~\ref{Rem:remarks_thmA}.
\end{Rem}

\begin{proof} We use identity $\L(\vep) = \mathrm{E}[\log (1 + \vep^2 X_\vep)]$ (Lemma~\ref{lem:existence_Xvep}) and expand the logarithm. There exists $C>0$ such that for all $x \geqslant 0$,
\begin{equation}
\left| \log(1+x) - \sum_{j =1}^K \frac{(-1)^{j+1}}{j} x^j \right| \leqslant C x^{\beta}.
\end{equation}
Consequently
\begin{equation}
\label{eq:proof_of_prop_regular_expansion_logarithm}
\L(\vep) = \sum_{j =1}^K \frac{(-1)^{j+1}}{j} \vep^{2 j} \mathrm{E}[X_\vep^j] + \O \left( \vep^{2 \beta} \mathrm{E}[X_\vep^{\beta}] \right).
\end{equation}

\begin{lem}
\label{lem:dvptpuissance}
For all $l \leqslant K$, the following expansion holds,
\begin{equation}
\label{eq:dvptpuissance}
\mathrm{E}[ X_\vep^l] = \sum_{k=0}^{K-l} (-1)^k g_{l, k}  \vep^{2 k} + \O(\vep^{2( \beta -l)} \mathrm{E}[X_\vep^\beta]),
\end{equation}
where, for all $ l \geqslant 1$ and $k \geqslant 0$, the coefficient $ g_{l,k}$ is a positive rational function of $\mathrm{E}[Z], \ldots, \mathrm{E}[Z^{l+k}]$.
\end{lem}

We first admit Lemma~\ref{lem:dvptpuissance} and conclude the proof of Proposition~\ref{prop:dvpt_lyapunov}. The substitution of~\eqref{eq:dvptpuissance} into~\eqref{eq:proof_of_prop_regular_expansion_logarithm} yields
\begin{equation}
\L(\vep) = \sum_{j=1}^K \sum_{k=0}^{K-j} \frac{(-1)^{j+k+1}}{j} \vep^{2 (j+k)} g_{j, k} + \O \left( \vep^{2 \beta} \mathrm{E}[X_\vep^{\beta}] \right).
\end{equation}
It can be rewritten
\begin{equation}
\L(\vep) = \sum_{s=1}^K (-1)^{s+1} \ell_s \vep^{2 s}  + \O \left( \vep^{2 \beta} \mathrm{E}[X_\vep^{\beta}] \right), \qquad \text{with} \qquad
\ell_s = \sum_{j=1}^K \sum_{k=0}^{K-j} \frac{  g_{j, k}}{j} \mathbf{1}_{j+k=s},
\end{equation}
and $\ell_s$ is a positive rational function of $\mathrm{E}[Z], \ldots, \mathrm{E}[Z^s]$ by inspection.
\end{proof}

We are left with the proof of Lemma~\ref{lem:dvptpuissance}, for which we briefly explain the strategy. Write the identity
\begin{equation}
\mathrm{E}[X_\vep^k] = \mathrm{E}[Z^k] \mathrm{E} \left[ \left( \frac{1+X_\vep}{1+ \vep^2 X_\vep}\right)^k \right].
\end{equation}
Then by expanding the denominator one gets
\begin{equation}
\label{eq:heuristique_dvpt_Xvep^k}
\mathrm{E}[X_\vep^k] = \mathrm{E}[Z^k] \sum_{j=0}^{n}  { -k \choose j} \vep^{2 j} \mathrm{E} \left[ (1+X_\vep)^k    X_\vep^j \right]+ \text{Remainder}.
\end{equation} 
It gives a relation between the moments of $X_\vep$ which will be used via a bootstrap procedure: the substitution of a regular expansion for $X_\vep$'s first moments into~\eqref{eq:heuristique_dvpt_Xvep^k} will provide a more precise expansion of $\mathrm{E}[X_\vep^k]$. That new expansion will in turn be injected into~\eqref{eq:heuristique_dvpt_Xvep^k} (for another $k$), to obtain a more precise regular expansion for that other moment, et cætera. 
Of course that procedure should be done in a specific order. Doing it rigorously will require a double induction, on $k$ and the length of the expansions. 
Let's now proceed to the detailed proof.

\begin{proof}[Proof of Lemma~\ref{lem:dvptpuissance}] 
Set $\delta = \beta -K$. We prove, using a course-of-values double induction with the lexicographic order on $(m,j)$, that if $j+m \leqslant K$, then $\mathrm{E}[X_\vep^j]$ has an expansion up to the order $\vep^{2 m}$: 
\begin{equation}
\label{eq:dvptpuissance_HR}
\mathrm{E}[ X_\vep^j] = \sum_{k=0}^{m} (-1)^k g_{j, k}  \vep^{2 k} + \O \left(\vep^{2( m+ \delta)} \mathrm{E}[X_\vep^\beta]\right),
\end{equation}
where for every $ j \geqslant 1$ and $k \geqslant 0$, the coefficient $ g_{j,k} $ is a positive rational function of $\mathrm{E}[Z], \ldots, \mathrm{E}[Z^{j+k}]$.
Of course $\mathrm{E}[X_\vep^0]$ admits such an expansion, up to any order. All that remains is the inductive step. Fix $l \geqslant 1$ and $n \geqslant 0$ such that $l+ n \leqslant K$ and suppose that~\eqref{eq:dvptpuissance_HR} holds
\begin{enumerate}[label=(\Alph*)]
\item  \label{item:proofdvptpuissance1} for all $j \leqslant K$ and $m \leqslant (n-1)\wedge (K -j)$;
\item \label{item:proofdvptpuissance2} for all $j \leqslant l-1$, and $m \leqslant n$.
\end{enumerate}
We want to show that it also holds for $(j,m)=(l,n)$. To this end, write
\begin{equation}
\label{eq:dvptXvep^l}
\mathrm{E}[X_\vep^l] = \mathrm{E} \left[ \left( Z \frac {1 + X_\vep}{1+\vep^2 X_\vep} \right)^l \right] = \mathrm{E}[Z^l] \sum_{r=0}^l {l \choose r} \mathrm{E} \left[ \frac{X_\vep^r}{(1+\vep^2 X_\vep)^l} \right].
\end{equation}
We want to expand the denominator with respect to $\vep$. Let $C>0$ be such that for any $x\geqslant 0$ and $l,m \leqslant K$,
\begin{equation}
\label{eq:dvpt_(1+x)^-l}
\left| \frac{1}{(1+ x)^l} - \sum_{i=0}^m {-l \choose i} x^i \right| \leqslant C x^{m+ \delta}.
\end{equation}
Thus, for every $r \leqslant l$,
\begin{equation}
\label{eq:dvptquotient}
\left| \mathrm{E}\left[\frac{X_\vep^r}{(1+ \vep^2 X_\vep)^l}\right] - \sum_{i=0}^n {-l \choose i} \vep^{ 2 i} \mathrm{E}[X_\vep^{i+r}]\right| \leqslant C \vep^{2(n+ \delta)}\mathrm{E}[X_\vep^{r+n+\delta}] 
\leqslant C \vep^{2(n+ \delta)} \max_{0 \leqslant k \leqslant K}\mathrm{E}[X_\vep^{k+\delta}].
\end{equation}
Actually
\begin{equation}
\label{eq:X_vep^beta}
\max_{k \leqslant K}\mathrm{E}[X_\vep^{k+\delta}] = \O(\mathrm{E}[X_\vep^\beta]).
\end{equation}
Indeed, if $1 \leqslant k \leqslant K-1$, then $\mathrm{E}[Z^{k+\delta}] <1$, so $\mathrm{E}[X_\vep^{k+\delta}]\leqslant \mathrm{E}[X_0^{k+\delta}] <+\infty$ (Lemmas~\ref{lem:X0} and~\ref{lem:existence_Xvep}). On the other hand $\mathrm{E}[X_\vep^{K+\delta}] =\mathrm{E}[X_\vep^\beta] \geqslant \mathrm{E}[Z^{\beta}] >0$ (Lemma~\ref{lem:existence_Xvep}).
Thus, with~\eqref{eq:dvptquotient} and~\eqref{eq:X_vep^beta}, we can write, for every $r \leqslant l$,
\begin{equation}
\mathrm{E}\left[\frac{X_\vep^r}{(1+ \vep^2 X_\vep)^l}\right] = \sum_{i=0}^n {-l \choose i} \vep^{ 2 i} \mathrm{E}[X_\vep^{i+r}] + \O(\vep^{2(n+ \delta)}\mathrm{E}[X_\vep^\beta]).
\end{equation}
And then, injecting it into~\eqref{eq:dvptXvep^l}, we get
\begin{equation}
\mathrm{E}[X_\vep^l]  = \mathrm{E}[Z^l] \sum_{r=0}^l {l \choose r} \sum_{i =0}^n {-l \choose i} \vep^{ 2 i} \mathrm{E} \left[ X_\vep^{i+r}\right] + \O(\vep^{2(n+\delta)}\mathrm{E}[X_\vep^\beta] ).
\end{equation}
We then isolate the term ``$(i,r)=(0,l)$'' --- that is $\mathrm{E}[Z^l] \mathrm{E}[X_\vep^l]$ --- on the left-hand side and divide by $1-\mathrm{E}[Z^l]$, to get
\begin{equation}
\label{eq:isolating_Xvep^l}
\mathrm{E}[X_\vep^l]= \frac{\mathrm{E}[Z^l]}{1-\mathrm{E}[Z^l]} \sum_{\substack{0 \leqslant r \leqslant l, \, 0 \leqslant i \leqslant n \\ (i,r) \neq (0,l)}} {l \choose r}  {-l \choose i} \vep^{ 2 i} \mathrm{E} \left[ X_\vep^{i+r}\right]+ \O(\vep^{2(n+\delta)} \mathrm{E}[X_\vep^\beta]).
\end{equation}
We claim that the induction hypothesis provides expansions for all these terms, up to the required order. The induction hypothesis~\eqref{eq:dvptpuissance_HR} on $\mathrm{E}[X_\vep^{i+r}]$ (induction hypothesis with $j=i+r$ and $m=n-i$, which is contained in the item~\ref{item:proofdvptpuissance2} if $i=0$ and in the item~\ref{item:proofdvptpuissance1} if $i \geqslant 1$), states that
\begin{equation}
\mathrm{E}[ X_\vep^{i+r}  ] =\sum_{k=0}^{n-i} \vep^{2 k} (-1)^k g_{i+r, k} + \O \left(\vep^{2( n-i + \delta)} \mathrm{E}[X_\vep^\beta]\right).
\end{equation} 
We then inject it into~\eqref{eq:isolating_Xvep^l}. It yields
\begin{equation}
\label{eq:demo_lem_dvptpuissance_étape_n-1}
 \mathrm{E}[X_\vep^l] = \frac{\mathrm{E}[Z^l]}{1-\mathrm{E}[Z^l]} \sum_{\substack{0 \leqslant r \leqslant l, \, 0 \leqslant i \leqslant n \\ (i,r) \neq (0,l)}}   {l \choose r}  {-l \choose i} \left(   \vep^{2 i} \sum_{k=0}^{n-i} \vep^{2 k} (-1)^k g_{i+r,k} +  \O \left(\vep^{2( n + \delta)} \mathrm{E}[X_\vep^\beta]\right) \right).
\end{equation}
One can already observe that it is a regular expansion of $\mathrm{E}[X_\vep^l]$ up to the order $n$, as expected. The following lines intend to derive a recursive formula for $g_{l,k}$ so as to check its sign. First note that
\begin{equation}
\label{eq:binomial_negatif}
{-l \choose i}=(-1)^i {l+i-1 \choose i}.
\end{equation}
Thus~\eqref{eq:demo_lem_dvptpuissance_étape_n-1} becomes
\begin{equation}
 \mathrm{E}[X_\vep^l] = \frac{\mathrm{E}[Z^l]}{1-\mathrm{E}[Z^l]} \sum_{\substack{0 \leqslant r \leqslant l, \, 0 \leqslant i \leqslant n \\ (i,r) \neq (0,l)\\ 0 \leqslant k \leqslant n-i}}   {l \choose r}  {l+i-1 \choose i}  \vep^{2 (k+i)} (-1)^{k+i} g_{i+r,k} +  \O \left(\vep^{2( n + \delta)} \mathrm{E}[X_\vep^\beta] \right).
\end{equation}
Eventually, it can be written as
\begin{equation}
\label{eq:defgls1}
 \mathrm{E}[X_\vep^l]=\sum_{s=0}^n  (-1)^s g_{l,s} \vep^{2  s} +  \O \left(\vep^{2( n + \delta)} \mathrm{E}[X_\vep^\beta] \right),
\end{equation}
with, for every $s \leqslant n$,
\begin{equation}
\label{eq:defgls2}
g_{l,s} = \frac{\mathrm{E}[Z^l]}{1-\mathrm{E}[Z^l]} \sum_{\substack{0 \leqslant r \leqslant l, \, 0 \leqslant i \leqslant n \\ (i,r) \neq (0,l)\\ 0 \leqslant k \leqslant n-i}}   {l \choose r}  {l+i-1 \choose i}  g_{i+r,k} \mathbf{1}_{i+k=s}.
\end{equation}
Thanks to the induction hypothesis, it is a positive rational function of $\mathrm{E}[Z], \ldots, \mathrm{E}[Z^{l+n}]$.
The inductive step is proved, and the lemma follows.
\end{proof}

\section{Theorem~\ref{The:dvpt_avec_erreur_general}: lower bound on the error}
\label{sec:error}

We prove here the lower bound on the error given in Theorem~\ref{The:dvpt_avec_erreur_general}, formula~\eqref{eq:The_encadrement_Rvep}. 
We already saw in Proposition~\ref{prop:dvpt_lyapunov}'s proof, when we studied the signs before the coefficients $\ell_k$ or $g_{l,k}$, that when expanding the algebraic fractions $(1+ \vep^2 X_\vep)^{-r}$, the term $\vep^{2 n}$ always comes with the sign $(-1)^n$. The same occurs for the error, at each step, at the order $\vep^{K+1}$: it comes with the sign $(-1)^{K+1}$. As a result, the error terms, which invariably accumulate with the same sign, effectively add up and cannot offset one another. In practice, these error terms can also be bounded from below. It yields the next result.

\begin{prop}
\label{prop:error_estimation}
Fix an integer $K \in \A \cup\{0\}$ and $B>0$. There exists $c>0$ such that, for all $\vep >0$,
\begin{equation}
(-1)^{K+2} \left[ \L(\vep) - \sum_{k=1}^K (-1)^{k+1} \ell_k \vep^{2 k} \right]  \geqslant c  \vep^{2(K+1)}\mathrm{E}[X_\vep^{K+1}\mathbf{1}_{\vep^2 X_\vep \leqslant B}],
\end{equation}
where the coefficients $(\ell_k)$ are the same as in Proposition~\ref{prop:dvpt_lyapunov}.
\end{prop}

Unsurprisingly, a similar scheme as in Proposition~\ref{prop:dvpt_lyapunov}'s proof will be used. We will proceed to a double induction, corresponding to an underlying bootstrap procedure. The only actual difference compared to Section~\ref{sec:regular} is that the estimate~\eqref{eq:dvpt_(1+x)^-l} is replaced by the lower bound
\begin{equation}
\frac{1}{(1+ x)^m} - \sum_{i=0}^{r}  {-m \choose i} x^i  \geqslant  C (- x)^{r+1} \mathbf{1}_{x \leqslant B}.
\end{equation}
We begin with the equivalent of Lemma~\ref{lem:dvptpuissance} in this new perspective.

\begin{lem} 
\label{lem:dvptpuissance_erreur}
Fix an integer $K \in \A \cup\{0\}$ and $B>0$. There exists $c>0$ such that for all $1 \leqslant l \leqslant K$, and $0 \leqslant n \leqslant K-l$, the following holds, for the same real coefficients~$(g_{l,k})$ as in Lemma~\ref{lem:dvptpuissance}
\begin{equation}
\mathrm{E}[ X_\vep^l] -\sum_{k=0}^{n} (-1)^k g_{l, k} \vep^{2 k}  \; \,
\begin{cases}
\geqslant  c (-1)^{n+1} \vep^{2(n+1)}\mathrm{E}[X_\vep^{l+n}\mathbf{1}_{\vep^2 X_\vep \leqslant B}] \qquad \text{if $n+1$ is even,}\\
\leqslant  c (-1)^{n+1} \vep^{2(n+1)}\mathrm{E}[X_\vep^{l+n}\mathbf{1}_{\vep^2 X_\vep \leqslant B}]
\qquad \text{if $n+1$ is odd.}
\end{cases}
\end{equation}
\end{lem}

\begin{proof}If $K=0$ the statement is empty, so suppose $K \geqslant 1$. 
It will be useful to recall formula~\eqref{eq:binomial_negatif}.
Fix $B >0$. There exists $C >0$ such that for all $1 \leqslant l \leqslant K+1$ and $n\leqslant K+1$, and for all $x \geqslant 0$,
\begin{equation}
\label{eq:encadrement_(1+x)^-l1}
\frac{1}{(1+ x)^l} - \sum_{i=0}^{n-1}  {l+i-1\choose i} (-x)^i  
\begin{cases}
\geqslant  C (- x)^{n} \mathbf{1}_{x \leqslant B} 
\qquad \text{if $n$ is even,}
\\
\leqslant  C (- x)^{n} \mathbf{1}_{x \leqslant B}
\qquad \text{if $n$ is odd.}
\end{cases}
\end{equation}

As in Lemma~\ref{lem:dvptpuissance}, we carry out a proof by course-of-values double induction. More precisely, set 
\begin{equation}
\tilde C := C \min_{ 1 \leqslant l \leqslant K} \frac{\mathrm{E}[Z^l]}{1-\mathrm{E}[Z^l]}.
\end{equation}
We prove that if $j \geqslant 1$, $m \geqslant 0$ and $j+m \leqslant K+1$ then
\begin{equation}
\label{eq:sing_HR}
   \mathrm{E}[ X_\vep^j  ] -\sum_{k=0}^{m-1} \vep^{2 k} (-1)^k g_{j, k}  \geqslant  \tilde C (-1)^m \vep^{2 m} \mathrm{E}[X_\vep^{j+m}\mathbf{1}_{\vep^2 X_\vep \leqslant B}]
\end{equation}
if $m$ is even; and the same with an inequality in the opposite direction if $m$ is odd. The base case $m=0$ is immediate. For the inductive step, we fix $l \geqslant 1$, $n \geqslant 1$ such that $l+n \leqslant K+1$ and we suppose that~\eqref{eq:sing_HR} holds for all $(j,m)$ with $m \leqslant n-1$ and $1 \leqslant j \leqslant K+1-m$, and for all $(j,n)$ with $1 \leqslant j \leqslant l-1$. We want to prove~\eqref{eq:sing_HR} for $(j,m)=(l,n)$. For the sake of simplicity, the proof will only be written for $n$ even (inequalities would be in the opposite direction if $n$ is odd). First write the identity
\begin{equation}
\label{eq:dvpt_sing_Xvep^l}
\mathrm{E}[X_\vep^l] = \mathrm{E} \left[ \left( Z \frac {1 + X_\vep}{1+\vep^2 X_\vep} \right)^l \right] = \mathrm{E}[Z^l] \sum_{r=0}^l {l \choose r} \mathrm{E} \left[ \frac{X_\vep^r}{(1+\vep^2 X_\vep)^l} \right].
\end{equation}
Using~\eqref{eq:encadrement_(1+x)^-l1} we get,
\begin{equation}
\mathrm{E}[X_\vep^l] \geqslant \mathrm{E}[Z^l] \sum_{r=0}^l  \left \{ {l \choose r} \sum_{i=0}^{n-1}  {l+i-1\choose i} (-1)^i \vep^{2 i}  \mathrm{E}[X_\vep^{i+r}] + C (-1)^n  \vep^{2 n} \mathrm{E}[X_\vep^{n+r}\mathbf{1}_{\vep^2 X_\vep \leqslant B}] \right \}.
\end{equation}
We subtract the term $\mathrm{E}[Z^l] \mathrm{E}[X_\vep^l]$ (term $(i,r)=(0,l)$) and divide by $1-\mathrm{E}[Z^l]$ (which is positive) to obtain
\begin{equation}
\label{eq:encadrement_sing_LB}
\begin{aligned}
 \mathrm{E}[X_\vep^l] \geqslant \frac{\mathrm{E}[Z^l]}{1-\mathrm{E}[Z^l]}  \sum_{\substack{0 \leqslant r \leqslant l, \,  0 \leqslant i \leqslant n-1 \\ (i,r) \neq (0,l)}}   {l \choose r} \Bigg \{ & {l+i-1\choose i} (-1)^i  \vep^{2 i}  \mathrm{E}[X_\vep^{i+r}] \\
 &+ C (-1)^n  \vep^{2 n} \mathrm{E}[X_\vep^{n+r}\mathbf{1}_{\vep^2 X_\vep \leqslant B}] \Bigg \}.
 \end{aligned}
\end{equation}
We use the induction hypothesis on $\mathrm{E}[X_\vep^{i+r}]$ (induction hypothesis~\eqref{eq:sing_HR} with $j=i+r$ and $m=n-i$), that is
\begin{equation}
\mathrm{E}[ X_\vep^{i+r}  ] -\sum_{k=0}^{n-i-1} \vep^{2 k} (-1)^k g_{i+r, k} \geqslant \tilde C (-1)^{n-i} \vep^{2 (n-i)}  \mathrm{E}[X_\vep^{r+n}\mathbf{1}_{\vep^2 X_\vep \leqslant B}] ,
\end{equation} 
if $n-i$ is even, and the opposite if it is odd. In any case, injecting these lower bounds into~\eqref{eq:encadrement_sing_LB} yields
\begin{equation}
\begin{aligned}
 \mathrm{E}[X_\vep^l] \geqslant \frac{\mathrm{E}[Z^l]}{1-\mathrm{E}[Z^l]} &\sum_{\substack{0 \leqslant r \leqslant l, \, 0 \leqslant i \leqslant n-1 \\ (i,r) \neq (0,l)}}   {l \choose r} \Bigg \{ {l+i-1\choose i} \Bigg( (-1)^i \vep^{2 i} \sum_{k=0}^{n-i-1} \vep^{2 k} (-1)^k g_{i+r,k} \\
 &+\tilde  C (-1)^n \vep^{2 n} \mathrm{E}[X_\vep^{n+r}\mathbf{1}_{\vep^2 X_\vep \leqslant B}] \Bigg) + C (-1)^n \vep^{2 n}\mathrm{E}[X_\vep^{n+r}\mathbf{1}_{\vep^2 X_\vep \leqslant B}] \Bigg \} .
\end{aligned}
\end{equation}
The first line corresponds to the regular part already found in Lemma~\ref{lem:dvptpuissance} equations~\eqref{eq:defgls1} and~\eqref{eq:defgls2}; the second line contains the $\vep^{2 n}$-terms which we want to bound from below:
\begin{equation}
 \mathrm{E}[X_\vep^l] \geqslant \sum_{s=0}^{n-1} \vep^{2 s} (-1)^k g_{l, s} +  (-1)^n \vep^{2 n} Q_n,
\end{equation}
with
\begin{equation}
Q_n =\frac{\mathrm{E}[Z^l]}{1-\mathrm{E}[Z^l]}  \sum_{\substack{0 \leqslant r \leqslant l, \, 0 \leqslant i \leqslant n-1 \\ (i,r) \neq (0,l)}}   {l \choose r} \Bigg \{ {l+i-1\choose i} \tilde  C +C \Bigg \} \mathrm{E}[X_\vep^{n+r}\mathbf{1}_{\vep^2 X_\vep \leqslant B}].
\end{equation}
Since all the terms in $Q_n$ are non-negative, it is larger than any of them
\begin{equation}
Q_n \geqslant \frac{\mathrm{E}[Z^l]}{1-\mathrm{E}[Z^l]} C \mathrm{E}\left[X_\vep^{n+l}\mathbf{1}_{\vep^2 X_\vep \leqslant B}\right] \geqslant \tilde C \mathrm{E}\left[X_\vep^{n+l}\mathbf{1}_{\vep^2 X_\vep \leqslant B}\right].
\end{equation}
This concludes the proof of the inductive step and thus the proof of the lemma.
\end{proof}

\begin{proof}[Proof of Proposition~\ref{prop:error_estimation}]
Let $c'=c'(B,K) >0$ be such that for all $x \geqslant 0$,
\begin{equation}
\log(1+x) \geqslant \sum_{l=1}^K \frac{(-x)^{l+1}}{l} + c' (-x)^{K+2}\mathbf{1}_{x\leqslant B}
\end{equation}
if $K$ is even; and the same with an inequality in the opposite direction if $K$ is odd. For the sake of simplicity we suppose that $K$ is even in what follows.
Writing $\L(\vep) = \mathrm{E} \log(1+ \vep^2 X_\vep)$, we get
\begin{equation}
 \L(\vep) \geqslant \sum_{l=1}^K \frac{(-1)^{l+1}}{l} \vep^{2{l}} \mathrm{E}[X_\vep^l] + c' (-1)^{K+2} \vep^{2(K+1)} \mathrm{E}[X_\vep^{K+1}\mathbf{1}_{\vep^2 X_\vep \leqslant B}],
\end{equation}
and Lemma~\ref{lem:dvptpuissance_erreur} provides a lower bound for each term in the sum: for every $1 \leqslant l \leqslant K$,
\begin{equation}
(-1)^{l+1} \vep^{2{l}} \mathrm{E}[X_\vep^l] \geqslant (-1)^{l+1}  \vep^{2{l}} \sum_{k=0}^{K-l} \vep^{2 k} (-1)^k g_{l, k} +  c (-1)^{l+1} (-1)^{K+1-l}  \vep^{2(K+1)} \mathrm{E}[X_\vep^{K+1}\mathbf{1}_{\vep^2 X_\vep \leqslant B}].
\end{equation}
The conclusion results from the latter two inequalities.
\end{proof}

\section{Limiting behaviour of \texorpdfstring{$X_\vep$'s}{} divergent moments}

\label{section:Xvep^K+1}

First note that Theorem~\ref{The:dvpt_avec_erreur_general} is an immediate consequence of Propositions~\ref{prop:dvpt_lyapunov} and~\ref{prop:error_estimation}.
The goal of this section is to obtain estimates of the error $R_K(\vep)$, for which we now have
\begin{equation}
c \vep^{2 (K+1)} \mathrm{E}[X_\vep^{K+1} \mathbf{1}_{ \vep^2 X_\vep \leqslant B}] \leqslant  R_{K}(\vep) \leqslant C_\beta \vep^{2 \beta} \mathrm{E}[X_\vep^{\beta}].
\end{equation}
In order to give explicit estimates of $R_K(\vep)$ in terms of powers of $\vep$, one needs to understand the limiting behaviour of $X_\vep$'s moments (or truncated moments). The issue was partially addressed by Corollary~\ref{corol:moments_Xvep_premier_resultat}, which pinpointed the regimes of convergence or divergence of these moments. Namely $\mathrm{E}[X_\vep^\gamma]$ is bounded as $\vep$ goes to $0$ if $\mathrm{E}[Z^\gamma]<1$ and diverges if $\mathrm{E}[Z^\gamma]\geqslant 1$. In the following section we address the issue of the divergence speed when $\mathrm{E}[Z^{\gamma}]\geqslant 1$.

The first paragraph, based on renewal theory results, describing the heavy tail of $X_0$, will provide upper bounds for $X_\vep$'s divergent moments. The second paragraph will give lower bounds for these moments under the restriction that $Z$ is bounded.

\subsection{Upper bounds}
\label{subsection:moments_upperbound}

We will need the following result, which combine results by H. Kesten and A. K. Grincevi\v{c}ius depending if $\log Z$ has an arithmetic support or not (see \cite[Theorems~1,~3]{Kevei_2017} for a review).

\begin{lem}
\label{lem:tail_X0}
If $\mathrm{E}[Z^\alpha \log_+ Z] <+\infty$, then, as $x$ goes to $+\infty$,
\begin{equation}
\label{eq:tail_X_0}
\mathrm{P}(X_0 \geqslant x) =\O( x^{-\alpha}).
\end{equation}
\end{lem}

It readily gives the next two results. They provide explicit upper bounds for the speed of divergence of $X_\vep$'s moments. If you believe Conjecture~\ref{conj}, these upper bounds (except the first one when $\mathrm{E}[Z^\alpha]<1$) are of the good order of $\vep$. The first one will be used for $\alpha \in \{1, 2 , \ldots\}$ whereas the second will be needed when $\alpha$ is not an integer.

\begin{lem}
\label{lem:moments_upperbound_alpha}
If~$\mathrm{E}[Z^\alpha \log_+ Z] <+\infty$, then, as $\vep$ goes to $0$,
\begin{equation}
\mathrm{E}[X_\vep^\alpha] =\O \left( \log (1/\vep) \right).
\end{equation}
\end{lem}

\begin{proof}
The identity $X_\vep \overset{\normalfont{(d)}}{=} Z \frac{1+X_\vep}{1 + \vep^2 X_\vep}$
yields, for $\gamma \geqslant 0$,
\begin{equation}
\begin{aligned}
\mathrm{E}[X_\vep^{\gamma}] = \mathrm{E}[Z^{\gamma}] \mathrm{E}\left[\left( \frac{1+X_\vep}{1 + \vep^2 X_\vep}\right)^{\gamma} \right]
&\leqslant \mathrm{E}[Z^{\gamma}] \mathrm{E}\left[\left((1+ X_\vep) \wedge \vep^{-2}\right)^{\gamma}\right] \\
\label{eq:majoration_EXvep^gamma1}
&\leqslant \mathrm{E}[Z^{\gamma}] \mathrm{E}\left[\left((1+ X_0) \wedge \vep^{-2}\right)^{\gamma}\right].
\end{aligned}
\end{equation}
It can be rewritten
\begin{equation}
\label{eq:majoration_EXvep^gamma2}
\mathrm{E}[X_\vep^{\gamma}] \leqslant \mathrm{E}[Z^{\gamma}] \left( \gamma   \int_0^{\vep^{-2}} x^{\gamma-1} \mathrm{P}(X_0 > x-1) \d x + \vep^{-2 \gamma } \mathrm{P}(X_0 \geqslant \vep^{-2}-1) \right).
\end{equation}
With $\gamma = \alpha$, Lemma~\ref{lem:tail_X0} gives upper bounds for these two terms:
\begin{equation}
\vep^{-2 \alpha} \mathrm{P}(X_0 \geqslant \vep^{-2}-1) =\O(1) \qquad \text{and} \qquad
\int_{0}^{\vep^{-2}} x^{\alpha-1}  \mathrm{P}(X_0 >x-1) \d x =\O( \log (1/\vep)).
\end{equation}
\end{proof}

\begin{lem}
\label{lem:moments_upperbound_kappa}
Fix $\gamma > \alpha$ and assume that $\mathrm{E}[Z^{\gamma}]$ is finite. Then, as $\vep$ goes to $0$,
\begin{equation}
\mathrm{E}[X_\vep^\gamma] =\O( \vep^{2\alpha- 2 \gamma}).
\end{equation}
\end{lem}

\begin{proof}
We reuse inequality~\eqref{eq:majoration_EXvep^gamma2}. Lemma~\ref{lem:tail_X0}, which applies here, yields
\begin{equation}
\vep^{-2 \gamma} \mathrm{P}(X_0 \geqslant \vep^{-2}-1) =\O( \vep^{2\alpha- 2 \gamma})
\qquad
\text{and}
\qquad
\int_{0}^{ \vep^{-2}} x^{\gamma-1}  \mathrm{P}(X_0 >x-1) \d x =\O( \vep^{2\alpha- 2 \gamma}).
\end{equation}
\end{proof}

\begin{Rem}
\label{Rem:proof_E[Zalpha]<1_EZlogZ}
If $\mathrm{E}[Z^\alpha \log_+ Z]<+\infty$, the same techniques yields $\mathrm{E}[(\vep^2 X_\vep)^{\alpha} \wedge (\vep^2 X_\vep)^{\lceil \alpha \rceil}]= o(\vep^{2 \alpha})$. So, with Remark~\ref{Rem:O(xalpha_xbeta)}, we get a proof of the result claimed in Remark~\ref{Rem:remarks_thmA}, item~\ref{item:rem_thmA_ZlogZ}. Similarly using again the upper bound on the error provided by Remark~\ref{Rem:O(xalpha_xbeta)}, one obtains~\eqref{eq:expansionLvep_The_entier_pathologique}. To this end, an alternative version of Lemma~\ref{lem:tail_X0} should be used: when $\mathrm{E}[Z^\alpha]<1$, and under some extra technical assumptions, one has $\mathrm{P}(X_0 \geqslant x) = o(x^{-\alpha})$ (see \cite[Theorem 1.3]{Kevei_2016} and \cite[Theorem 8]{Kevei_2017}).
\end{Rem}

\subsection{Lower bounds when \texorpdfstring{$Z$}{} is bounded}
\label{subsection:moments_lowerbound}

We start with a quite general, albeit quite complex, lower bound for $X_\vep$'s moments.

\begin{lem}
\label{lem:lowerboundgamma}
Fix $\gamma \geqslant 1$, $C>0$, $B>1$ and $N \in \N$ and set $\tau = \frac{  \gamma }{B-1}\left(\frac{B}{B-1} +C \right)$. One has
\begin{equation}
\label{eq:lem_lowerboundgamma}
\mathrm{E}[X_\vep^\gamma] \geqslant \sum_{k=1}^N \mathrm{E}\left[Z^\gamma \mathbf{1}_{Z \leqslant B} \right]^{k} \exp \left( - \tau \vep^2  B^{k} \right) \mathrm{P}(X_\vep \leqslant C).
\end{equation}
\end{lem}

\begin{proof} 
Let $X_\vep^{(N)}$ be a copy of $X_\vep$ and $(Z_k)$ be iid copies of $Z$, independent of $X_\vep^{(N)}$. Define recursively, for $0 \leqslant k \leqslant N-1$,
\begin{equation}
\label{eq:Xvep^(k)}
X_\vep^{(k)} = Z_{k+1} \frac{1+X_\vep^{(k+1)}}{1+\vep^2 X_\vep^{(k+1)}}.
\end{equation}
For every $k \leqslant N$, one has $X_\vep^{(k)} \overset{\normalfont{(d)}}{=} X_\vep$.
On the other hand, one can derive the following lower bounds:
\begin{equation}
(X_\vep^{(0)})^\gamma=  Z_{1}^\gamma \frac{(1+X_\vep^{(1)})^\gamma}{(1+\vep^2 X_\vep^{(1)})^\gamma}  \geqslant  \frac{Z_{1}^\gamma}{(1+\vep^2 X_\vep^{(1)})^\gamma} +  \frac{Z_{1}^\gamma}{(1+\vep^2 X_\vep^{(1)})^\gamma} (X_\vep^{(1)})^\gamma .
\end{equation}
Here the condition $\gamma \geqslant 1$ is used through the convexity inequality $(1+x)^\gamma \geqslant 1 + x^\gamma$. Then, inductively,
\begin{equation}
(X_\vep^{(0)})^\gamma \geqslant \frac{Z_{1}^\gamma}{(1+\vep^2 X_\vep^{(1)})^\gamma} +  \frac{Z_{1}^\gamma}{(1+\vep^2 X_\vep^{(1)})^\gamma} \frac{Z_{2}^\gamma}{(1+\vep^2 X_\vep^{(2)})^\gamma} + \cdots + \prod_{j=1}^N \frac{Z_j^\gamma}{(1+\vep^2 X_\vep^{(j)})^\gamma} .
\end{equation}
By taking the expectation we get
\begin{equation}
\label{eq:minorationEXalpha}
\mathrm{E}[X_\vep^\gamma]  \geqslant \sum_{k=1}^{N} \mathrm{E} \left[  \prod_{j=1}^k \frac{Z_j^\gamma}{(1+\vep^2 X_\vep^{(j)})^\gamma} \right].
\end{equation}
If $X_\vep^{(k)} \leqslant C$ and $Z_k \leqslant B$, then, with definition~\eqref{eq:Xvep^(k)}, $X_\vep^{(k-1)} \leqslant B(1+C)$.
So, inductively, if $X_\vep^{(k)} \leqslant C$, and $Z_0, \ldots, Z_k \leqslant B$, then, for every $j \leqslant k$,
\begin{equation}
X_\vep^{(k-j)}\leqslant \sum_{i=1}^j B^i + B^{j} C \leqslant B^j \left(\frac{B}{B-1} +C \right)=B^j \sigma,
\end{equation}
with $\sigma=\frac{B}{B-1} +C$.
Thus,
\begin{equation}
\label{eq:fractionBsigma}
\prod_{j=1}^k \frac{Z_j^\gamma}{(1+\vep^2 X_\vep^{(j)})^\gamma} 
\geqslant \left[\prod_{j=1}^k \frac{Z_j^\gamma \mathbf{1}_{Z_j \leqslant B} }{(1+ \sigma \vep^2  B^{k-j})^{\gamma}} \right] \mathbf{1}_{X_\vep^{(k)} \leqslant C}  .
\end{equation}
We compute
\begin{equation}
 \prod_{j=1}^k \frac{1}{(1+ \sigma \vep^2  B^{k-j})^{\gamma}}  \geqslant \exp \left( - \sigma \gamma \vep^2  \sum_{j=1}^k B^{k-j} \right) \geqslant \exp \left( - \frac{\sigma  \gamma }{B-1} \vep^2  B^{k} \right) =\exp \left( - \tau \vep^2  B^{k} \right).
\end{equation}
Taking the expectation in~\eqref{eq:fractionBsigma} and using that $X_\vep^{(k)}$ and $Z_1, \ldots, Z_k$ are independent, we obtain
\begin{equation}
\mathrm{E} \left[  \prod_{j=1}^k \frac{Z_j^\gamma}{(1+\vep^2 X_\vep^{(j)})^\gamma} \right] 
\geqslant \mathrm{E}\left[Z^\gamma \mathbf{1}_{Z \leqslant B} \right]^{k} \exp \left( - \tau \vep^2  B^{k} \right) \mathrm{P}(X_\vep \leqslant C).
\end{equation}
The conclusion follows by injecting this lower bound into~\eqref{eq:minorationEXalpha}.
\end{proof}

One could expect to use that general lower bound for any given $Z$. However it only gives satisfactory results when $Z$ is bounded. In that case we can get rid of the indicator $\mathbf{1}_{Z \leq B}$ in~\eqref{eq:lem_lowerboundgamma}.
\newline

\begin{lem}
\label{lem:Xvep_borné}
If $Z$ has a bounded support then $X_\vep \leqslant \vep^{-2} \|Z \|_{L^\infty}$ almost surely.
\end{lem}

\begin{proof} It is an immediate consequence of the invariance identity $X_\vep \overset{\normalfont{(d)}}{=} Z \frac{1+ X_\vep}{1+ \vep^2 X_\vep}$ and of the inequality~$\frac{1+ x}{1+ \vep^2 x} \leqslant \vep^{-2}$, which holds for every $x \geqslant 0$.
\end{proof}

Lemma~\ref{lem:Xvep_borné} justifies that we only study $X_\vep$'s moments instead of its truncated moments: as long as $B$ is chosen larger than $\|Z\|_{L^\infty}$ one has
\begin{equation}
\mathrm{E}\left[X_\vep^{K+1} \mathbf{1}_{ \vep^2 X_\vep \leqslant B}\right] = \mathrm{E}\left[X_\vep^{K+1}\right].
\end{equation}
In the next two lemmas we give a lower bound for $X_\vep$'s moments when $Z$ is bounded. In that instance, note that $\mathrm{E}[Z^\alpha]=1$: the set $\A$ cannot takes the form $\A = (0, \alpha]$.
Lemma~\ref{lem:moments_lowerbound_entier} will be used if $\alpha$ is an integer, and Lemma~\ref{lem:moments_lowerbound_non_entier} when $\alpha$ is not an integer. However, both of them hold true regardless of the nature of $\alpha$.

\begin{lem} 
\label{lem:moments_lowerbound_entier}
For $\alpha \geqslant 1$, if $Z$ has a bounded support then, for some $c>0$, and $\vep$ sufficiently small,
\begin{equation}
\mathrm{E}[X_\vep^\alpha] \geqslant c \log (1/\vep) .
\end{equation}
\end{lem}

\begin{proof} Recall that since $Z$ is bounded, $\mathrm{E}[Z^\alpha]=1$. Choose $\gamma = \alpha$ and $B = \|Z\|_\infty$ in Lemma~\ref{lem:lowerboundgamma} to get
\begin{align}
\mathrm{E}[X_\vep^\alpha]  \geqslant \sum_{k=1}^N \mathrm{E}\left[Z^\alpha \right]^{k} \exp \left( - \tau \vep^2  B^{k} \right) \mathrm{P}(X_\vep \leqslant C) 
\geqslant N  \exp \left( - \tau \vep^2  B^{N}\right) \mathrm{P}(X_\vep \leqslant C).
\end{align}
First note that, thanks to Lemma~\ref{lem:convergence_en_loi_Xvep},
\begin{equation}
\mathrm{P}(X_\vep \leqslant C) \longrightarrow \mathrm{P}(X_0 \leqslant C),
\end{equation}
which is positive if $C$ is large enough. Choosing $N= N_\vep = \lfloor 2 \frac{1}{\log B} \log \frac{1}{\vep} \rfloor$, we obtain
\begin{equation}
\mathrm{E}[X_\vep^\alpha] \geqslant N_\vep \exp (- \tau ) \mathrm{P}(X_\vep \leqslant C) \geqslant c  \log (1/\vep).
\end{equation}
\end{proof}

\begin{Rem}
If $Z$ is not bounded but $\mathrm{E}[Z^{\kappa}] <+\infty$ for some $\kappa > \alpha$ then, with another choice of $B_\vep$ and $N_\vep$, one can get the slightly weaker lower bound
$\mathrm{E}[X_\vep^\alpha] \geqslant c \frac{\log (1/\vep)}{\log \log (1/\vep)}$.
\end{Rem}

\begin{lem}
\label{lem:moments_lowerbound_non_entier}
If $Z$ is bounded, and if $\gamma \geqslant 1$ is such that $\mathrm{E}[Z^\gamma]>1$, then, for some $c>0$, and for $\vep$ sufficiently small,
\begin{equation}
\mathrm{E}[X_\vep^\gamma] \geqslant c \vep^{-2 \eta},
\qquad \text{ where} \qquad
\eta = \frac{\log \mathrm{E}[Z^\gamma]}{\log \|Z\|_\infty} \in (0,\gamma - \alpha).
\end{equation}
\end{lem}

\begin{proof} Set $B = \|Z\|_{L^\infty}$ in Lemma~\ref{lem:lowerboundgamma} to get
\begin{align}
\mathrm{E}[X_\vep^\gamma]  \geqslant \sum_{k=1}^N \mathrm{E}\left[Z^\gamma \right]^{k} \exp \left( - \tau \vep^2  B^{k} \right) \mathrm{P}(X_\vep \leqslant C) \geqslant \mathrm{E}\left[Z^\gamma \right]^{N} \exp \left( - \tau \vep^2  B^{N}\right) \mathrm{P}(X_\vep \leqslant C).
\end{align}
Choosing again $N_\vep = \lfloor 2 \frac{1}{\log B} \log \frac{1}{\vep} \rfloor$, we obtain
$
\mathrm{E}[X_\vep^\gamma] \geqslant \mathrm{E}\left[Z^\gamma \right]^{N_\vep} \exp \left( - \tau \right) \mathrm{P}(X_\vep \leqslant C) \geqslant c  \vep^{- 2 \eta}.
$
\end{proof}

\subsection{Proof of Theorem~\ref{The:thm_concret}} 
\label{section:proof_theorem}

\begin{proof}[\unskip\nopunct]
We recall here the upper and lower bounds provided by Theorem~\ref{The:dvpt_avec_erreur_general}:
\begin{equation}
\label{eq:proof_the_encadrement_Rvep}
c \vep^{2 (K+1)} \mathrm{E}[X_\vep^{K+1} \mathbf{1}_{ \vep^2 X_\vep \leqslant B}] \leqslant  R_{K}(\vep) \leqslant C_\beta \vep^{2 \beta} \mathrm{E}[X_\vep^{\beta}],
\end{equation}
If $\alpha=+\infty$ then $R_K(\vep) \leqslant C_{K+1} \vep^{2(K+1)} \mathrm{E}[X_\vep^{K+1}] \leqslant \vep^{2(K+1)} \mathrm{E}[X_0^{K+1}]$ (Lemma~\ref{lem:existence_Xvep}). Since $\mathrm{E}[X_0^{K+1}]$ is finite (Lemma~\ref{lem:X0}), the result~\eqref{eq:expansionLvep_the_alpha_infini} follows.

From now on we suppose that $\alpha$ is finite and $\mathrm{E}[Z^\alpha]=1$ and we set $K = \lceil \alpha \rceil-1$. By Corollary~\ref{corol:moments_Xvep_premier_resultat},
\begin{equation}
R(\vep) \geqslant c \vep^{2 (K+1)} \mathrm{E}[X_\vep^{K+1} \mathbf{1}_{ \vep^2 X_\vep \leqslant B}] \gg \vep^{2 (K+1)}.
\end{equation}
If $\alpha$ is an integer then the lower and upper bounds given by~\eqref{eq:The_encadrement_Rvep_entier} or~\eqref{eq:The_encadrement_Rvep_entier_Z_borne} follow from Lemmas~\ref{lem:moments_upperbound_alpha} and~\ref{lem:moments_lowerbound_entier}.
If $\alpha$ is not an integer then the lower and upper bounds~\eqref{eq:The_encadrement_Rvep_nonentier} given by Theorem~\ref{The:thm_concret} are a consequence of Lemmas~\ref{lem:moments_upperbound_kappa} (with $\gamma$ such that $\mathrm{E}[Z^\gamma] <+\infty$) and~\ref{lem:moments_lowerbound_non_entier} (with $\gamma = K+1$).
\end{proof}

\appendix
\section{Generalization to higher dimension}
\label{appendix:generalization}

The techniques developed in the previous sections are sufficiently robust to be used in more general settings. We apply them to a square matrix of size $d+1$ which is a perturbation of a matrix alike $\Diag(1,Z)$, which still have a preferred direction. Since the proofs are only slightly different from the previous sections, they will be only sketched in this appendix. We will just point out the arguments that must be adapted and many details will be omitted.
\newline

We now consider the $(d+1) \times (d+1)$ matrix
\begin{equation}
M_\vep = \begin{pmatrix}
1 & \vep L_\vep \\ \vep C_\vep & N_\vep
\end{pmatrix},
\end{equation}
where $L_\vep$ and $C_\vep$ are random vectors of size $d$, and $N_\vep$ is a random matrix, of size $d \times d$.
We are still interested in the Lyapunov exponent, defined by the limit
\begin{equation}
\L (\vep) = \lim_{n \to + \infty} \frac{1}{n} \log \left \| M_{n,\vep} \cdots M_{1,\vep}\right \|,
\end{equation}
where $(M_{k,\vep})_{k \geqslant 1}$ are iid copies of $M_\vep$. This limit exists almost surely and is deterministic (see again~\cite{Furstenberg_Kesten_60}) as soon as for every $\vep >0$, $\mathrm{E}[\log_+ \|M_\vep\|] < + \infty$.

We derive in this section a regular expansion for $\L(\vep)$, alike the expansion provided by Proposition~\ref{prop:dvpt_lyapunov} in the previous setting. However, no lower bound on the error will be given here. We start by deriving a formula alike ``$\L(\vep)= \mathrm{E}[\log(1+\vep^2 X_\vep)]$'' (Lemma~\ref{lem:expressionlyapunov_generalized}).
\newline

In the whole section $\| \cdot \|$ will denote a given norm on $\R^d$ or $\R^{d+1}$, as well as the induced operator norm on $\M_d(\R)$ or $\M_{d+1}(\R)$. On another note, if $x,y \in \R^d$, we will write $x \leqslant y$ if the inequality holds coordinatewise. Similarly the stochastic dominance $\preccurlyeq$ will be extended to random vectors: $X \preccurlyeq Y$ means that there exists a copy $\tilde X$ of $X$ and a copy $\tilde Y$ of $Y$ satisfying $\tilde X \leqslant \tilde Y$ almost surely (coordinatewise).
\newline

Let's introduce the assumptions under which we will work in the section. Observe that under these assumptions, the condition $\mathrm{E}[\log_+ \|M_\vep\|] < + \infty$ is fulfilled so the Lyapunov exponent is well defined.

\begin{Hyp}
\label{hyp:generalized}
We assume that the following holds, for every $\vep \in (0, \vep_0)$.
\begin{enumerate}[label=(\alph*)]
\item \label{hyp:generalized_irred} The random matrix $M_\vep$ has non-negative entries. And, almost surely, there exists $N\geqslant 1$ such that the product $M_{N,\vep} \cdots M_{1,\vep}$ has positive entries.
\item \label{hyp:generalized_Nvep_Cvep} There exists $\delta_\vep >0$ such that $\mathrm{E}[\| N_\vep \|^{\delta_\vep}]<1$ and $\mathrm{E}[\|C_\vep\|^{\delta_\vep}]<+\infty$.
\item \label{hyp:generalized_Lvep} $\mathrm{E}[\log_+ \|L_\vep\|] <+ \infty$.
\end{enumerate}
\end{Hyp}

Before deriving the formula for the Lyapunov exponent, we introduce the random vector $Y_\vep$, which will play the same role as $X_0$ in our new setting (except that here it will depend on $\vep$). Namely it will be used through stochastic dominances.

\begin{lem}
\label{lem:def_Yvep}
Fix $\vep \in (0,\vep_0)$ and let $(N_{\vep,k}, C_{\vep,k})$ be iid copies of $(N_\vep, C_\vep)$. The series
\begin{equation}
\label{eq:def_Yvep}
Y_\vep = \sum_{n=0}^{+\infty} N_{\vep,1} \ldots N_{\vep,n-1} C_{\vep,n}
\end{equation}
converges almost surely. Moreover $\mathrm{E}[\log_+ \|Y_\vep\|]$ is finite. If, in addition,
\begin{equation}
\label{eq:control_EXvep^beta_generalized}
\limsup_{\vep \to 0} \mathrm{E}[\|N_\vep\|^\beta] <1 \qquad \text{and} \qquad \limsup_{\vep \to 0} \mathrm{E}[\|C_\vep\|^\beta] <+\infty,
\end{equation}
then $\mathrm{E}[\|Y_\vep\|^\beta] = \O(1)$ as $\vep $ goes to $0$.
\end{lem}

\begin{proof}
Since all the entries of $M_\vep$ are non-negative, the sum~\eqref{eq:def_Yvep} is always defined. A priori, some of its entries could be $+ \infty$. Denote by $Y_\vep$ the random vector defined by this infinite sum. Using Minkowski's inequality or another convexity inequality as for Lemma~\ref{lem:X0}, one proves, under Assumption~\ref{hyp:generalized}~\ref{hyp:generalized_Nvep_Cvep}, that $\mathrm{E}[\|Y_\vep\|^{\delta_\vep}]$ is finite. So $Y_\vep$'s entries are almost surely finite. With the same technique, we prove the rest of the lemma.
\end{proof}

The next lemma provides the desired formula for $\L(\vep)$.

\begin{lem} 
\label{lem:expressionlyapunov_generalized}
There exists a random vector $X_\vep \in \R^d$, with non-negative entries, satisfying
\begin{equation}
\begin{pmatrix}
1 \\ \vep X_\vep 
\end{pmatrix}
\overset{\normalfont{(d)}}{=}
\begin{pmatrix}
1 & \vep L_\vep \\ \vep C_\vep & N_\vep
\end{pmatrix} \begin{pmatrix}
1 \\ \vep X_\vep 
\end{pmatrix} \qquad \text{ in the projective space } \P^{d}(\R),
\end{equation}
or equivalently,
\begin{equation}
\label{eq:point_fixe_generalized}
X_\vep \overset{\normalfont{(d)}}{=}  \frac{C_\vep+ N_\vep X_\vep}{1+ \vep^2 L_\vep X_\vep},
\end{equation}
where $C_\vep, N_\vep$ and $L_\vep$ are the blocks of the random matrix $M_\vep$, independent of $X_\vep$.
One has $X_\vep \preccurlyeq Y_\vep$. Moreover the Lyapunov exponent can be written as
\begin{equation}
\label{eq:expression_lyapunov_generalization}
\L(\vep) = \mathrm{E}[\log(1+\vep^2 L_\vep X_\vep)] .
\end{equation}
And for every $\vec x, \vec y \in \R_+^{d+1}$, 
\begin{equation}
\label{eq:limite_coeffs}
\L(\vep) = \lim_{n \to  \infty} \frac{1}{n} \log \left \langle \vec x , M_{n,\vep} \cdots M_{1,\vep}  \vec y \right \rangle.
\end{equation}
\end{lem}

\begin{proof}
The method is the same as in Lemma~\ref{lem:existence_Xvep}'s proof for $2 \times 2$ matrices. We fix iid copies $(M_{\vep,n})$ of $M_\vep$ and set $x_0 = 0_{\R^d}$. Then define inductively, for $n \geqslant 0$, the random variables
\begin{equation}
x_{n+1} = \frac{C_{\vep,n}+ N_{\vep,n} x_n}{1+ \vep^2 L_{\vep,n} x_n}.
\end{equation}
Observe that since all the vectors have non-negative entries, one can write, coordinatewise,
\begin{equation}
x_{n+1} \leqslant C_{\vep,n}+ N_{\vep,n} x_n.
\end{equation}
So, by an easy induction, $x_n \preccurlyeq Y_\vep$ for every $n \geqslant 0$.
The end of the proof is the same as for Lemma~\ref{lem:existence_Xvep}. We do not reiterate all the details here. Just note that we do not claim the uniqueness of a non-negative solution to~\eqref{eq:point_fixe_generalized} and that Assumption~\ref{hyp:generalized}~\ref{hyp:generalized_irred} is a sufficient condition for H. Hennion's result to apply.
\end{proof}

To state our main result, and more precisely to formulate its premises, some multi-index notations will be required, which we set in the next lines. The norm of a multi-index $\bm{\lambda} \in \N^d$ will be denoted by~$|\bm{\lambda}|$:
\begin{equation}
|\bm{\lambda}|:=\lambda_1 + \ldots \lambda_{d}.
\end{equation}
For every $l \geqslant 0$, there are ${ l+d-1 \choose d-1}$ multi-indices with norm $l$: it is the number of (weak) compositions of $l$ into $d$ non-negative integers. For a vector $x \in \R^d$ and a multi-index $\bm{\lambda} \in \N^d$, we define the multi-index power
\begin{equation}
x^{\bm{\lambda}} = x_1^{\lambda_1} \times \cdots \times x_d^{\lambda_d}.
\end{equation}
Similarly, for a matrix $A \in \M_d(\R)$ and a multi-index $\bm{\omega} \in \N^{d^2} \simeq \M_d(\N)$, define
\begin{equation}
A^{\bm{ \omega}} = \prod_{i,j} (A_{i,j})^{\omega_{i,j}} \qquad \text{and} \qquad |\bm{\omega}| =  \sum_{i,j} \omega_{i,j}.
\end{equation}
There should be no confusion with a standard matrix power since $\bm{\omega}$ is a multi-index.

For $l \geqslant 0$, consider the square matrix $G^{(l)}$ with size ${ l+d-1 \choose d-1}$, whose elements are
\begin{equation}
G^{(l)}_{\bm{\lambda},\bm{\lambda'}} = \sum_{\substack{\bm{\omega} \in \N^{d^2} \\ \sum_j \omega_{i,j} = \lambda_i \\ \sum_i \omega_{i,j} = \lambda'_j}} \lim_{\vep \to 0} \mathrm{E}\left[ N_\vep^{\bm{\omega}}\right], \qquad \text{for } \bm{\lambda}, \bm{\lambda'} \in \N^{d} \, \text{  such that } |\bm{\lambda}|= |\bm{\lambda'}|=l.
\end{equation}
Note that all the multi-indices $\bm{\omega}$ in the sum have norm $|\bm{\omega} |=l$.
The matrix $G^{(l)}$ will play a similar role as $\mathrm{E}[Z^l]$ in this generalized context. Of course these matrices, which require the existence of $\lim_{\vep \to 0}\mathrm{E}\left[ N_\vep^{\bm{\omega}}\right]$, are not always defined.
\newline

We have set enough notations to state the generalization of Proposition~\ref{prop:dvpt_lyapunov}, giving a regular expansion of the Lyapunov exponent $\L(\vep)$.

\begin{prop}
\label{prop:dvpt_lyapunov_generalized}
Fix $K \geqslant 0$ and $\beta \in (K,K+1]$. Suppose that
\begin{enumerate}
\item For all multi-indices $\bm{\lambda},\bm{\mu} \in \N^{d}$, $\bm{\omega} \in \N^{d^2}$ such that $l=|\bm{\lambda}|+|\bm{\mu}|+|\bm{\omega}| \leqslant K$, $\mathrm{E}[L_\vep^{\bm{ \lambda}} C_\vep^{\bm{\mu}} N_\vep^{\bm{ \omega}}]$ is finite and admits a regular expansion, as $\vep$ goes to~0, up to the order $2(K-l)$:
\begin{equation}
\label{eq:dvpt_lyapunov_generalized_hyp1}
\mathrm{E}[L_\vep^{\bm{ \lambda}} C_\vep^{\bm{\mu}} N_\vep^{\bm{ \omega}}] = \sum_{r=0}^{2(K-l)} c_{\bm{\lambda},\bm{\mu},\bm{\omega},r} \vep^r + \O(\vep^{2(\beta-l)});
\end{equation}
\item For all $1 \leqslant l \leqslant K$, the matrix $I-G^{(l)}$ is invertible;
\item $\limsup_{\vep \to 0} \mathrm{E}[\|L_\vep\|^\beta]$ is finite.
\end{enumerate}
Then there exist real coefficients $q_2, \ldots q_{2 K}$ such that, as $\vep$ goes to $0$,
\begin{equation}
\L (\vep) = \sum_{k=2}^{2 K} q_k \vep^{ k} +\O( \vep^{2 \beta} \mathrm{E}[1+\|X_\vep\|^\beta ]).
\end{equation}
\end{prop}

\begin{Rem}For Proposition~\ref{prop:dvpt_lyapunov_generalized} to be usable, one needs to control $\mathrm{E}[\|X_\vep\|^\beta]$. With Lemmas~\ref{lem:def_Yvep} and~\ref{lem:expressionlyapunov_generalized}, one has $\mathrm{E}[\|X_\vep\|^\beta] = \O(1)$ as $\vep$ goes to~0 as soon as~\eqref{eq:control_EXvep^beta_generalized} hold.
\end{Rem}

\begin{Rem}
One could be surprised that the upper bound involves $\mathrm{E}[1+\|X_\vep\|^\beta ]$ instead of $\mathrm{E}[\|X_\vep\|^\beta ]$. Such a caution was not necessary in the previous context since the latter was bounded form below as $\vep$ goes to $0$. Here, a priori, it could happen that $\mathrm{E}[\|X_\vep\|^\beta ]$ vanishes as $\vep$ goes to $0$.
\end{Rem}

\begin{Rem}
The existence of $G^{(l)}$, for $l \leqslant K$, is ensured by the assumption~\eqref{eq:dvpt_lyapunov_generalized_hyp1}, which gives $\lim_{\vep \to 0} \mathrm{E}[N_\vep^{\bm{\omega}}] = c_{\bm{0}, \bm{0}, \bm{\omega}, 0}$.
The invertibility of $I-G^{(l)}$ is the counterpart of the assumption ``$\mathrm{E}[Z^l] <1$'' in Proposition~\ref{prop:dvpt_lyapunov}.
\end{Rem}

\begin{proof}
The same proof as for Proposition~\ref{prop:dvpt_lyapunov} works: one expands the logarithm:
\begin{equation}
\begin{aligned}
\label{eq:lyapunov_firstexpansion_generalized}
\mathrm{E}[ \log(1+\vep^2 & L_\vep X_\vep)]  = \sum_{k=0}^K \frac{(-1)^{k+1}}{k} \vep^{2 k} \mathrm{E}[(L_\vep X_\vep)^k] + \O(\vep^{2 \beta}  \mathrm{E}[(L_\vep X_\vep)^\beta])\\
&
 = \sum_{k=0}^K \frac{(-1)^{k+1}}{k} \vep^{2 k}  \sum_{1 \leqslant r_1, \ldots, r_k \leqslant d} \mathrm{E}\left[ \prod_{i=1}^k L^{(r_i)}_\vep \right] \mathrm{E}\left[ \prod_{i=1}^k X^{(r_i)}_\vep \right] +\O(\vep^{2 \beta}  \mathrm{E}[(L_\vep X_\vep)^\beta]),
 \end{aligned}
\end{equation}
where $x^{(r)}$ stands for the $r^\text{th}$ coordinate of $x$. Note that
\begin{equation}
\mathrm{E}[(L_\vep X_\vep)^\beta] \leqslant \mathrm{E}[\|L_\vep\|^\beta] \mathrm{E}[\|X_\vep\|^\beta] \leqslant C \mathrm{E}[\|X_\vep\|^\beta],
\end{equation}
and that for any $r_1, \ldots, r_k$ there exists $\bm{\lambda} \in \N^d$, with norm $k$ such that
$
\mathrm{E}\left[ \prod_{i=1}^k X^{(r_i)}_\vep \right] = \mathrm{E}[ X_{\vep}^{\bm{\lambda}}].
$
Thus we need expansions for $X_\vep$'s moments. They are given by the next lemma. By substituting the regular expansion~\eqref{eq:dvptpuissance_generalization}, given in Lemma~\ref{lem:dvptpuissance_generalization}, in the expansion~\eqref{eq:lyapunov_firstexpansion_generalized} of $\L(\vep)$, the proof of Proposition~\ref{prop:dvpt_lyapunov_generalized} will be complete.
\end{proof}

\begin{lem}
\label{lem:dvptpuissance_generalization}
Under Proposition~\ref{prop:dvpt_lyapunov_generalized}'s premises, for all $l \leqslant K$, and $\bm{\lambda}\in \N^{d}$, such that $|\bm{\lambda}|  = l$, the following expansion holds, for some real coefficients $(g_{\bm{\lambda},k})$:
\begin{equation}
\label{eq:dvptpuissance_generalization}
\mathrm{E}[ X_{\vep}^{\bm{\lambda}}] = \sum_{k=0}^{2(K-l)} \vep^{k} g_{\bm{\lambda}, k} + \O( \vep^{2 (\beta-l)} \mathrm{E}[1+\|X_\vep\|^\beta ]).
\end{equation}
\end{lem}

\begin{proof}[Sketch of proof of Lemma~\ref{lem:dvptpuissance_generalization}] 
We can follow the same proof as for Lemma~\ref{lem:dvptpuissance}. We go back to that proof to understand how the present one must be adjusted. The only point which merits special attention is the line~\eqref{eq:isolating_Xvep^l} where the term $\mathrm{E}[Z^l] \mathrm{E}[X_\vep^l]$ is isolated on the left-hand side. That line could be summarized as follow: we wrote
\begin{equation}
\mathrm{E}[X_\vep^l] = \mathrm{E}[Z^l] \mathrm{E}[X_\vep^l] +(\diamondsuit_l),
\end{equation}
where $(\diamondsuit_l)$ stands for all the terms in the expansion of $\mathrm{E}[X_\vep^{l}]$ for which the induction hypothesis provided an expansion up to the required order. To be explicit,
\begin{equation}
(\diamondsuit_l) = \mathrm{E}[Z^l] \sum_{\substack{0 \leqslant j \leqslant l, \, 0 \leqslant i \leqslant n \\ (i,j) \neq (0,l)}} {l \choose j}  {-l \choose i} \vep^{ 2 i} \mathrm{E} \left[ X_\vep^{i+j}\right]+ \O(\vep^{2(n+\delta)} \mathrm{E}[X_\vep^\beta]).
\end{equation}
Then we could conclude by writing
\begin{equation}
\mathrm{E}[X_\vep^l] = \frac{1}{1-\mathrm{E}[Z^l]} (\diamondsuit_l),
\end{equation}
and applying the induction hypothesis. That is where was used the condition ``$\mathrm{E}[Z^l]<1$'' (actually $\mathrm{E}[Z^l] \neq 1$ was enough), and this is where will be used the invertibility of $1-G^{(l)}$. 

In our generalized setting, we still carry out an induction on $(n, l= |\lambda|)$ (equipped with the lexicographic order). For the inductive step, there are a lot of multi-indices with given norm $l$. They will be solved simultaneously, by writing a joint system satisfied by all these multi-indices moments $\mathrm{E}[X_\vep^{\bm{\lambda}}]$ with $|\bm{\lambda}|=l$. To this end, use the identity
\begin{align}
\mathrm{E}[X_\vep^{\bm{\lambda}}] =\mathrm{E} \left[ \left(\frac{C_\vep+ N_\vep X_\vep}{1+ \vep^2 L_\vep X_\vep} \right)^{\bm{\lambda}} \right]=\mathrm{E} \left[ \frac{(C_\vep+ N_\vep X_\vep)^{\bm{\lambda}}}{(1+ \vep^2 L_\vep X_\vep)^l} \right].
\end{align}
Then develop the denominator
\begin{equation}
\mathrm{E}[X_\vep^{\bm{\lambda}}]=\mathrm{E}\left[ (C_\vep+ N_\vep X_\vep)^{\bm{\lambda}} \left( \sum_{j=0}^n  {-l \choose j} \vep^{2 j} (L_\vep X_\vep)^j + \vep^{2(n+\delta)} O( (L_\vep X_\vep)^{n+\delta}) \right) \right].
\end{equation}
Eventually, after manipulation, that moment takes the form
\begin{equation}
\mathrm{E}[X_\vep^{\bm{\lambda}}] = \sum_{ \bm{\lambda'} : | \bm{\lambda'}|=l} G^{(l)}_{\bm{\lambda},\bm{\lambda'}}\mathrm{E}[X_\vep^{\bm{\lambda'}}] +(\diamondsuit_{\bm{\lambda}}),
\end{equation}
where, again, $(\diamondsuit_{\bm{\lambda}})$ stands for all the term in the expansion of $\mathrm{E}[X_\vep^{\bm{\lambda}}]$ for which the induction hypothesis, and the premise~\eqref{eq:dvpt_lyapunov_generalized_hyp1} of Proposition~\ref{prop:dvpt_lyapunov_generalized}, provide an expansion up to the required order. Then, since $I-G^{(l)}$ is invertible, one can solve that joint system satisfied by the family $(\mathrm{E}[X_\vep^{\bm{\lambda}}])$:
\begin{equation}
\mathrm{E}[X_\vep^{\bm{\lambda}}] = \left[(I-G^{(l)})^{-1} (\diamondsuit)\right] _{\bm{\lambda}} = \sum_{ \bm{\lambda'} : | \bm{\lambda'}|=l} \left((I- G^{(l)})^{-1}\right)_{\bm{\lambda},\bm{\lambda'}}(\diamondsuit_{\bm{\lambda'}}).
\end{equation}
That concludes the proof of the induction step and thus the proof of the lemma.
\end{proof}

\begin{Rem}
The same methods as in Section~\ref{sec:error} can produce the lower bound on the error
\begin{equation}
(-1)^{K+2} R_K(\vep) \geqslant c \vep^{2(K+1)}  \mathrm{E} \left[ (L_\vep X_\vep)^{K+1} \mathbf{1}_{L_\vep X_\vep \leqslant B} \right] + O(\vep^{2(K+1)}) ,
\end{equation}
as long as~\eqref{eq:dvpt_lyapunov_generalized_hyp1} holds with $\beta = K+1$.
\end{Rem}

\paragraph{Application to a 1D Ising model} The product of random matrices considered in the first sections appeared in \cite{Derrida_Hilhorst_83} to express the free energy of the nearest-neighbour Ising model on the line with inhomogeneous magnetic field. The generalization considered in this appendix allows finite range interactions to be included. Let us be more precise.
Consider the Ising model on $\T_N:=\Z / N \Z$, with homogeneous interactions up to the distance $d$ and inhomogeneous magnetic field $(h_k)$. It is the spin model with configurations\footnote{We choose $\{0,1\}^{\T_N}$ instead of $\{-1,1\}^{\T_N}$ to simplify the formulas. They are equivalent by easy manipulations.} ${\underline{\sigma} \in \{0,1\}^{\T_N}}$ whose Hamiltonian is
\begin{equation}
H(\underline{\sigma}) = \sum_{k \in \T_N} \left( h_k \sigma_k + \sum_{l=1}^d \alpha_l \mathbf{1}_{\sigma_k \neq \sigma_{k+l}} \right).
\end{equation}
The magnetic field $(h_k)_{k \in \T_N}$ is supposed to be iid. 
Thanks to a transfer matrix approach, the free energy in the thermodynamic limit can be expressed through a random matrix products:
\begin{equation}
f (T) = \lim_{N \to + \infty} \frac{1}{N} \log \Tr \left( \prod_{n=1}^N A_n \right),
\end{equation}
where $A_n$ is a $2^{d} \times 2^{d}$ sparse matrix (two non-zero entries on each line and each column) whose entries are the following. If $\underline{\tau},\underline{\upsilon} \in \{0,1\}^d$, which represent the partial configuration $(\sigma_n, \ldots, \sigma_{n+d-1})$ and its shift $(\sigma_{n+1}, \ldots, \sigma_{n+d})$, then
\begin{equation}
A_n (\underline{\tau},\underline{\upsilon}) = \exp \left( - \frac{1}{T} \tau_1 h_n - \frac{1}{T}  \sum_{l=1}^d \alpha_l \tau_l \upsilon_l \right)\mathbf{1}_{\tau_2 = \upsilon_1, \ldots, \tau_{d-1}=\upsilon_d}.
\end{equation}
One can check that Assumption~\ref{hyp:generalized}~\ref{hyp:generalized_irred} holds with $N=d$. Proposition~\ref{prop:dvpt_lyapunov_generalized} provides an expansion for the free energy $f(T)$ when the coupling constants $\alpha_l$ tend to be very large.
Set $Z_n = \exp(-h_n/T)$ and $\vep_l = \exp(-  \alpha_l/T)$ for every $l \leqslant d$. The parameters $\vep_l$ vanish when the coupling constants $\alpha_l$ tend to be very large. Then $A_n$ is a random perturbation of $\Diag(1, 0, \ldots, 0, Z_n)$ if one writes the configurations $\underline{\tau},\underline{\upsilon}$ in lexicographic order. Thus, Proposition~\ref{prop:dvpt_lyapunov_generalized} yields
\begin{equation}
f(T) = \sum_{\bm{\lambda} \in \N^d : |\bm{\lambda}| < \beta} c_{\bm{\lambda}} \vep_1^{\lambda_1} \cdots \vep_d^{\lambda_d} + \O \left( \sum_{l=1}^d \vep_l^{\beta} \right), 
\end{equation}
as soon as $\mathrm{E}[Z^\beta] <1$ (note that $\beta$ is not the inverse temperature here).

\begin{Rem}
Similarly, the results apply for an Ising model on a strip of finite width $s$ (i.e. $[N] \times [s]$), or a cylinder ($[N] \times \Z/s \Z$) with an inhomogeneous magnetic field and finite-range interactions, with free, fixed or periodic boundary conditions.
\end{Rem}

\paragraph{Acknowledgments} This work is part of my PhD thesis supervised by Giambattista Giacomin. I would like to thank him for giving me the opportunity to work on this subject and for fruitful discussions.

\small
\bibliographystyle{myplain}
\bibliography{biblioDHalpha}
%
%

%

\end{document}